\documentclass[a4paper,UKenglish]{lipics}
 
\usepackage{microtype}

\usepackage{stmaryrd}
\usepackage{pgf}
\usepackage{tikz}

\newtheorem{proposition}{Proposition}
\newtheorem{observation}{Observation}

\title{On terminating improvement in two-player games}

\author[1]{St\'ephane Le Roux}
\affil[1]{D\'epartement d'informatique, Universit\'e libre de Bruxelles\\
  Belgique\\
  \texttt{Stephane.Le.Roux@ulb.ac.be}}

\Copyright[nc-nd]
          {St\'ephane Le Roux}

\subjclass{I.2.11}
\keywords{Nash equilibrium, structures, potential games, weakly terminating games}

\serieslogo{}
\volumeinfo
  {Billy Editor, Bill Editors}
  {2}
  {Conference title on which this volume is based on}
  {1}
  {1}
  {1}
\EventShortName{}
\DOI{10.4230/LIPIcs.xxx.yyy.p}

\begin{document}

\maketitle

\begin{abstract}
A real-valued game has the finite improvement property (FIP), if starting from an arbitrary strategy profile and letting the players change strategies to increase their individual payoffs in a sequential but non-deterministic order always reaches a Nash equilibrium. \textit{E.g.}, potential games have the FIP. Many of them have the FIP by chance, though, since modifying even a single payoff may ruin the property. This article characterises (in quadratic time) the class of the finite games where the FIP not only holds but is also preserved when modifying all the occurrences of an arbitrary payoff. The characterisation relies on a pattern-matching sufficient condition for games (finite or infinite) to enjoy the FIP, and is followed by an inductive description of this class.

A real-valued game is weakly acyclic if the improvement described above can reach a Nash equilibrium. This article characterises the finite such games using Markov chains and almost sure convergence to equilibrium. It also gives an inductive description of the two-player such games. 
 \end{abstract}

\section{Introduction}

Game theory is the theory of competitive interactions between agents who have different interests. To describe how such real-world systems may stabilise, game theory especially relies on the notions of game and Nash equilibrium (NE for short), popularised by Nash~\cite{Nash50}. As defined in~\cite{SLR14} and below, a game structure is a function from a Cartesian product; and a game is a game structure plus one binary relation assigned to each component of the product. 

\begin{definition}[Game structures and games in normal form]\label{defn:gsnf}
Game structures in normal form are tuples $\langle A,(S_a)_{a\in A},O,v\rangle$ satisfying the following: 
\begin{itemize}
\item $A$ is a non-empty set (of players, or agents),
\item $\prod_{a\in A}S_a$ is a non-empty Cartesian product (whose elements are the strategy profiles and where $S_a$ represents the individual strategies available to player $a$),
\item $O$ is a non-empty set (of possible outcomes),
\item $v:\prod_{a\in A} S_a\to O$ (the outcome function that values the strategy profiles),
\end{itemize}
A game in normal form is a tuple $\langle G,(\prec_a)_{a\in A}\rangle$ where $G$ is a game structure in normal form, and each $\prec_a$ is a binary relation over $O$ (modelling the preference of player $a$).
\end{definition}

The following definitions are meant to stress that the Nash equilibria of a game are the terminal states of a dynamic, distributed process where the players improve upon their outcome in a non-deterministic and asynchronous way.

\begin{definition}[Convertibility, preference over profiles, improvement, and Nash equilibrium]\label{defn:asyn-improv}\hfill Let $\langle A,(S_a)_{a\in A},O,v,(\prec_a)_{a\in A}\rangle$ be a game in normal form and let $S:=\prod_{a\in A} S_a$.
\begin{itemize}
\item For $s,s' \in S$ let $s\stackrel{c}{\twoheadrightarrow}_as'$ denote the ability of Player $a$ to convert $s$ to $s'$ by changing her own individual strategy, formally $s\stackrel{c}{\twoheadrightarrow}_as':=\forall b\in A-\{a\},\,s_b=s'_b$. 

\item Let $s\prec_a s'$ denote $v(s)\prec_a v(s')$, to refer also to the induced preference over the profiles.

\item Let $\twoheadrightarrow_a\,:=\,\prec_a\cap\stackrel{c}{\twoheadrightarrow}_a$ be the individual improvement reductions of the players and let $\twoheadrightarrow\,:=\,\cup_{a\in A}\twoheadrightarrow_a$ be the collective improvement reduction.

\item A strategy profile $s \in S$ is a Nash equilibrium if the formula $\forall a\in A,\forall s'\in S,\,\neg(s \twoheadrightarrow_a s')$ holds, \textit{i.e.} if $s$ is a terminal state of the collective improvement $\twoheadrightarrow$.
\end{itemize}
\end{definition}

Let us now describe two "good reasons" for a game to have a Nash equilibrium, as opposed to having NE "by chance". The first reason is to have a suitable underlying structure: the first game below involves players $a$ and $b$ with strategy sets $\{a_1,a_2\}$ and $\{b_1,b_2\}$, respectively. Its strategy profile $(a_2,b_2)$ is an NE, since $a$ and $b$ cannot obtain payoffs greater than $1$ and $2$, respectively. It happens by chance, though: rewriting the payoff pair $(1,2)$ into $(1,0)$ produces a game without NE; whereas the second game (in extensive form) has an NE for structural reasons: rewriting arbitrarily the payoffs at the leaves of the tree still produces a game with NE, as proved by Kuhn in~\cite{Kuhn53}.  The underlying structures with this property were fully characterised in~\cite{SLR14}: if a two-player game structure always produces a determined game when its outcomes are arbitrarily instantiated with $(1,0)$ and $(0,1)$, it always produces a game with NE when arbitrarily instantiated with (finitely many) real-valued payoffs. 

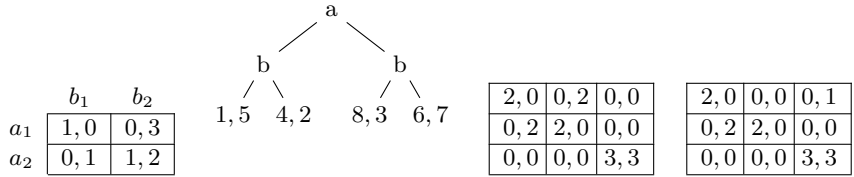
\begin{tabular}{cccc}
$\begin{array}{c|c|c|c|}
	\multicolumn{1}{c}{}&
	  \multicolumn{1}{c}{b_1}&
	  \multicolumn{1}{c}{b_2}\\
	  \cline{2-3}
 	  a_1 &  1,0 & 0,3 \\
	  \cline{2-3}
	   a_2 & 0,1 & 1,2 \\
	  \cline{2-3}
	 \end{array}$
&
\begin{tikzpicture}[level distance=7mm]
\node{a}[sibling distance=18mm]
	child{node{b}[sibling distance=8mm]
		child{node{$1,5$}}
		child{node{$4,2$}}
	}
	child{node{b}[sibling distance=8mm]
		child{node{$8,3$}}
		child{node{$6,7$}}
	};
\end{tikzpicture}
&
$\begin{array}{|c@{\;\vline\;}c@{\;\vline\;}c|}
	\cline{1-3}
	 2,0 & 0,2 & 0,0\\
	\cline{1-3}
	 0,2 & 2,0 & 0,0\\
	\cline{1-3}
	 0,0 & 0,0 & 3,3\\
	\cline{1-3}
\end{array}$
&
$\begin{array}{|c@{\;\vline\;}c@{\;\vline\;}c|}
	\cline{1-3}
	 2,0 & 0,0 & 0,1\\
	\cline{1-3}
	 0,2 & 2,0 & 0,0\\
	\cline{1-3}
	 0,0 & 0,0 & 3,3\\
	\cline{1-3}
\end{array}$

\end{tabular}

\bigskip

The second good reason for a game to have NE is the finite improvement property (FIP, \textit{i.e.} $\twoheadrightarrow$ is terminating), as defined and studied by Monderer and Shapley in~\cite{MS96}. Note that the third game above has one NE (bottom-right corner), although $\twoheadrightarrow$ is not terminating when starting in the $2\times 2$ upper-left corner. It is terminating in the fourth game, though.

However, again, rewriting the top $(0,0)$ into $(0,2)$ in the fourth game above ruins the FIP, which is an invitation to combine the two "good reasons" discussed above into preserved FIP when rewriting payoffs. This article deals with the two-player case only. To that purpose, the subgames of a game in normal form are defined below by restriction of the outcome function $v$ and of the preferences to Cartesian subsets of the set of profiles.

\begin{definition}[Subgame]\label{def:subgame}
A subgame of a game $\langle A,(S_a)_{a\in A},O,v,(\prec_a)_{a\in A}\rangle$ in normal form is a game $\langle A,(S'_a)_{a\in A},O,v\mid_{S'},(\prec_a\mid_{S'\times S'})_{a\in A}\rangle$ where $S'_a\subseteq S_a$ and $S':=\prod_{a\in A}S'_a$. Forgetting about the preferences yields the similar notion of sub game structure.

For a two-player game with $A=\{a,b\}$, a pattern (resp. rectangle) is a subgame (resp. sub game structure) where $|S_a|=|S_b|=2$. 
\end{definition}

Section~\ref{sect:cife} defines generic patterns: two are strongly forbidden, one weakly forbidden. Then it shows first that the absence of these patterns in finite two-player games guarantees the FIP, and second that the absence of the strongly forbidden ones guarantees that the collective maximising, \textit{i.e.} best-response-like, improvement terminates or passes near an NE. In the first case, examples suggest that the condition is almost necessary. The first case is then invoked to prove that assigning arbitrary acyclic preferences to a given game structure always yields a terminating $\twoheadrightarrow$ iff the structure has no forbidden rectangle, \textit{i.e.} every rectangle has an edge with a repeated outcome. Examples from Section~\ref{sect:afi-3p} then suggest that the three-player case is much harder, and Section~\ref{sect:ei-ene-ig} shows that the absence of the three patterns in infinite two-player games with continuous payoffs on a compact set of profiles guarantees termination of the collective $\epsilon$-improvement. Then Section~\ref{sect:id-mwfs} gives two inductive descriptions of the two-player game structures without forbidden rectangles.

Let us give three remarks. First, some real-valued games without forbidden patterns are neither exact, nor weighted, nor ordinal potential games as defined by Monderer and Shapley~\cite{MS96}, and conversely. Second,  Section~\ref{sect:cife} could be useful for mechanism design, the art of making the players behave well. Indeed, Section~\ref{sect:cife} ensures that a distributed system based on an appropriate game structure will evolve (quickly) towards an NE regardless of the actual preferences of the players. Third, the FIP is a very strong property that most of the games do not enjoy. Even convergence towards mixed NE is usually not guaranteed, but Hart and Mas-Colell~\cite{HM00} showed convergence towards mixed correlated equilibrium in finite real-valued games, a concept that is more general than that of NE and that was introduced by Aumann~\cite{Aumann87}. On the contrary, the potential games and especially this article focus on games that guarantee convergence in finite time or almost surely towards a (non-mixed) NE.

Finally, there is a more general good reason than FIP for games to have NE: in weakly acyclic games (see \cite{Young93}) NE are $\twoheadrightarrow$-reachable from any profile. Whereas \cite{MAS07} designs a sophisticated stochastic process converging almost surely to NE in such games,  Section~\ref{sect:wtmi} shows that weak termination in finite games is equivalent to almost sure convergence of a memoryless Markov process to an NE. Further, \cite{FJS10} proved that existence of a unique NE in every subgame implies weak acyclicity; also about subgames, Section~\ref{sect:wtmi} shows that weak termination of the maximising improvement is preserved when (recursively) removing a well-chosen strategy.

\section{Forbidden patterns and terminating improvement in 2-player games}\label{sect:cife}

Definition~\ref{def:el-pat} below defines eight generic patterns that are meant to match or not to match the patterns of a given game. Note that these patterns do not constitute a comprehensive account of the patterns that may occur in an abstract two-player game in normal form, but they are the ones that matter as far as this article is concerned. Also, given a preference $\prec_a$, the notation $x \preceq_a y$ stands for $\neg(y \prec_a x)$ and the notation $x \sim_a y$ for $x \preceq_a y \,\wedge\, y\preceq_a x$.

\begin{definition}[Forbidden patterns]\label{def:el-pat}
Given two players $a$ and $b$ and their preferences $\prec_a$ and $\prec_b$, the patterns below are defined up to permutation of the strategies and the players. In the remainder of this article, these patterns may be referred to by their full names or by the initials of the names only. The first two patterns (EC and PUB) are called the strongly forbidden patterns, and the third pattern (PUC) is called the weakly forbidden pattern.

\bigskip

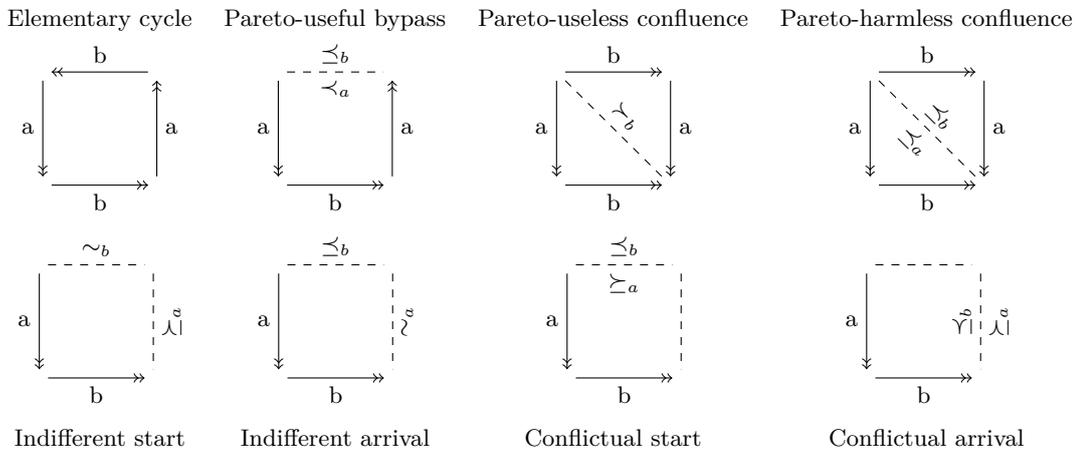
\begin{tabular}{cccc}
Elementary cycle & Pareto-useful bypass & Pareto-useless confluence & Pareto-harmless confluence\\
\begin{tikzpicture}[node distance=1.5cm]
  \node(x){};
  \node(y)[below of = x]{};
  \node(z)[right of = y]{};
  \node(t)[right of = x]{};

  \draw [->>] (x) to node[left] {a} (y);
  \draw [->>] (y) to node[below] {b} (z);
  \draw [->>] (t) to node[above] {b} (x);
  \draw [->>] (z) to node[right] {a} (t);
  \end{tikzpicture}
&
\begin{tikzpicture}[node distance=1.5cm]
  \node(x){};
  \node(y)[below of = x]{};
  \node(z)[right of = y]{};
  \node(t)[right of = x]{};

  \draw [->>] (x) to node[left] {a} (y);
  \draw [->>] (y) to node[below] {b} (z);
  \draw [->>] (z) to node[right] {a} (t);
  \draw [dashed] (x) to node[above] {$\preceq_b$} node[below] {$\prec_a$}(t);
  \end{tikzpicture}
&
\begin{tikzpicture}[node distance=1.5cm]
  \node(x){};
  \node(y)[below of = x]{};
  \node(z)[right of = y]{};
  \node(t)[right of = x]{};

  \draw [->>] (x) to node[left] {a} (y);
  \draw [->>] (y) to node[below] {b} (z);
  \draw [->>] (x) to node[above] {b} (t);
  \draw [->>] (t) to node[right] {a} (z);
  \draw [dashed] (x) to node[above, rotate=-45] {$\succ_b$} (z);
  \end{tikzpicture}
&
\begin{tikzpicture}[node distance=1.5cm]
  \node(x){};
  \node(y)[below of = x]{};
  \node(z)[right of = y]{};
  \node(t)[right of = x]{};

  \draw [->>] (x) to node[left] {a} (y);
  \draw [->>] (y) to node[below] {b} (z);
  \draw [->>] (x) to node[above] {b} (t);
  \draw [->>] (t) to node[right] {a} (z);
  \draw [dashed] (x) to node[above, rotate=-45] {$\preceq_b$} node[below, rotate=-45] {$\preceq_a$}(z);
  \end{tikzpicture}
\\
\begin{tikzpicture}[node distance=1.5cm]
  \node(x){};
  \node(y)[below of = x]{};
  \node(z)[right of = y]{};
  \node(t)[right of = x]{};

  \draw [->>] (x) to node[left] {a} (y);
  \draw [->>] (y) to node[below] {b} (z);
  \draw [dashed] (x) to node[above] {$\sim_b$} (t);
  \draw [dashed] (t) to node[below, rotate=90] {$\succeq_a$} (z);
  \end{tikzpicture}
&
\begin{tikzpicture}[node distance=1.5cm]
  \node(x){};
  \node(y)[below of = x]{};
  \node(z)[right of = y]{};
  \node(t)[right of = x]{};

  \draw [->>] (x) to node[left] {a} (y);
  \draw [->>] (y) to node[below] {b} (z);
  \draw [dashed] (x) to node[above] {$\preceq_b$} (t);
  \draw [dashed] (t) to node[below, rotate=90] {$\sim_a$} (z);
  \end{tikzpicture}
&
\begin{tikzpicture}[node distance=1.5cm]
  \node(x){};
  \node(y)[below of = x]{};
  \node(z)[right of = y]{};
  \node(t)[right of = x]{};

  \draw [->>] (x) to node[left] {a} (y);
  \draw [->>] (y) to node[below] {b} (z);
  \draw [dashed] (x) to node[above] {$\preceq_b$} node [below] {$\succeq_a$} (t);
  \draw [dashed] (t) -- (z);
  \end{tikzpicture}
&
\begin{tikzpicture}[node distance=1.5cm]
  \node(x){};
  \node(y)[below of = x]{};
  \node(z)[right of = y]{};
  \node(t)[right of = x]{};

  \draw [->>] (x) to node[left] {a} (y);
  \draw [->>] (y) to node[below] {b} (z);
  \draw [dashed] (t) to node[above, rotate=90] {$\preceq_b$} node[below, rotate=90] {$\succeq_a$} (z);
  \draw [dashed] (x) -- (t);
  \end{tikzpicture}
\\
Indifferent start & Indifferent arrival & Conflictual start & Conflictual arrival
\end{tabular}
\end{definition}

Observation~\ref{obs:comp-pat} below is a tool to make case disjunctions, and is proved in the appendix. 

\begin{observation}\label{obs:comp-pat}
Given a game with the two players $a$ and $b$, each rectangle that has a path of $\twoheadrightarrow$ of length at least two along the edges of the rectangle (like in the left-hand picture just below) actually matches (up to symmetry) one of the eight patterns from Definition~\ref{def:el-pat}.
\end{observation}

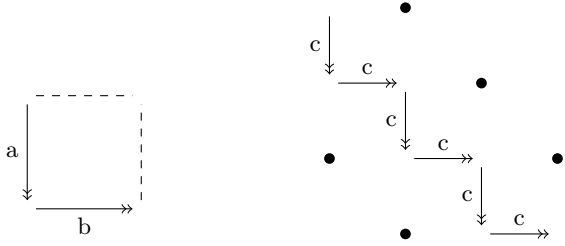
\begin{tabular}{c@{\hspace{2cm}}c}
\begin{tikzpicture}[node distance=1.5cm]
  \node(x){};
  \node(y)[below of = x]{};
  \node(z)[right of = y]{};
  \node(t)[right of = x]{};

  \draw [dashed] (x) -- (t);
  \draw [dashed] (z) -- (t);
  \draw [->>] (x) to node[left] {a} (y);
  \draw [->>] (y) to node[below] {b} (z);
  \end{tikzpicture}

&
\begin{tikzpicture}[node distance=1cm]
  \node(y0){};
  \node(x1)[below of = y0]{};
  \node(y1)[right of = x1]{};
  \node(x2)[below of = y1]{};
  \node(y2)[right of = x2]{};
  \node(x3)[below of = y2]{};
  \node(y3)[right of = x3]{};
  \node(z0)[right of = y0]{};
  \node(z1)[below of = x1]{};
  \node(z2)[right of = y1]{};
  \node(z3)[below of = x2]{};
  \node(z4)[right of = y2]{};

  \fill (z0) circle (2pt);
  \fill (z1) circle (2pt);
  \fill (z2) circle (2pt);
  \fill (z3) circle (2pt);
  \fill (z4) circle (2pt);

  \draw [->>] (y0) to node[left] {c} (x1);
  \draw [->>] (x1) to node[above] {c} (y1);
  \draw [->>] (y1) to node[left] {c} (x2);
  \draw [->>] (x2) to node[above] {c} (y2);
  \draw [->>] (y2) to node[left] {c} (x3);
  \draw [->>] (x3) to node[above] {c} (y3);
  \end{tikzpicture}
\end{tabular}

Definition~\ref{def:path-sheath} below introduces the notion of a sheath of a (convertibility) path. It is not yet clear how useful it is for Lemma~\ref{lem:fp-decrease}, but it is definitely useful for Lemma~\ref{lem:fp-max-impr} and Theorem~\ref{thm:sfp-ne}. The right-hand picture above illustrates Definition~\ref{def:path-sheath}: a sheath made of a six-step path plus five generated bullet points. (The letter $c$ stands for convertibility.)

\begin{definition}[Sheath of a convertibility path]\label{def:path-sheath}
Let $s_1\stackrel{c}{\twoheadrightarrow} s_2\stackrel{c}{\twoheadrightarrow}\dots\stackrel{c}{\twoheadrightarrow} s_n$ be a (player-alternating) convertibility path in a two-player game $\langle\{a,b\},S_a,S_b,O,v,\prec_a,\prec_b\rangle$, and let us define the \emph{sheath} of the path as the set of profiles $\{s_1,\dots,s_n\} \cup \{(s^a_i,s^b_{i+2})\,\mid\,1\leq i\leq n-2\,\wedge\,s_i\twoheadrightarrow_as_{i+1}\twoheadrightarrow_b s_{i+2}\}\cup\{(s^a_{i+2},s^b_{i})\,\mid\,1\leq i\leq n-2\,\wedge\,s_i\twoheadrightarrow_bs_{i+1}\twoheadrightarrow_a s_{i+2}\}$.
\end{definition}

Lemma~\ref{lem:fp-decrease} below means that if a path of collective improvement is (locally) minimal in length among all the paths from the same starting profile to the same target profile, some measure is decreasing along the path. Recall that $x \preceq_a y$ stands for $\neg(y \prec_a x)$.

\begin{lemma}\label{lem:fp-decrease}
Let be a game with the two players $a$ and $b$, and let us assume that the preferences satisfy the following for $P \in \{a,b\}$.
\[\begin{array}{l@{\hspace{1cm}}r}
\forall x,y,z,\quad x\prec_P y\,\wedge\,y\prec_P z\,\Rightarrow\, x\prec_P z & \mbox{(transitivity)}\\
\forall x,y,z,\quad x\prec_P y\,\wedge\,y\preceq_P z\,\Rightarrow\, x\prec_P z & \mbox{(pseudo-transitivity)}\\
\forall x,y,z,\quad x\preceq_P y\,\wedge\,y\prec_P z\,\Rightarrow\, x\prec_P z & \mbox{(pseudo-transitivity)}
\end{array}\]
\noindent Let $s_0\twoheadrightarrow_b s_1\twoheadrightarrow_a s_2 \twoheadrightarrow_b s_3\twoheadrightarrow_a s_4\twoheadrightarrow_b s_5$ be a path where the $s_i$ are pairwise disjoint, and let us assume that there is (within the sheath of the path) no shorter path from $s_0$ to $s_5$ that involves the two players. Also assume that no forbidden pattern occurs within the sheath. Then $v(s_3)\prec_a v(s_4)\preceq_a v(s_1)$. 
\end{lemma}

\begin{proof}
Let us first prove that $\preceq_P$ is transitive for $P \in \{a,b\}$. Let $x \preceq_P y \preceq_P z$. If $z \prec_P x$, then pseudo-transitivity and $y\preceq_P z$ implies $y \prec_P x$, contradiction. The path looks like the first picture below up to row and column permutation (and repetition) of the game matrix. By Observation~\ref{obs:comp-pat} and the assumption that there is no forbidden pattern within the sheath, the rectangle that is determined by the profiles $s_1$, $s_2$, and $s_3$ may match five different patterns. By (pseudo-)transitivity of the preferences, the first three patterns (PHC, IS, and IA) contradict minimality of the path, as shown by the additional improvement arrows in the second, third, and fourth pictures below. If it matches the CS pattern, as in the fifth picture below, the longer improvement arrow starting at $s_0$ holds by pseudo-transitivity; the longer dashed inequality holds not to contradict minimality of the path; and preference (pseudo-)transitivity then yields the required inequalities. For the CA pattern let us make a case disjunction on the rectangle determined by $s_2$, $s_3$ and $s_4$. If it matches PHC, IS or IA, it contradicts minimality of the path, just as before, so only the PHC-case is displayed. If it matches CA, the longer improvement arrow holds by pseudo-transitivity; the longer dashed inequality holds not to contradict minimality of the path; and preference (pseudo-)transitivity then yields the required inequalities. Finally let us assume that it matches CS. The two longer improvement arrows holds by pseudo-transitivity; because of them the two longer dashed inequalities follow, not to contradict minimality of the path; and then because of these inequalities the two shorter additional improvement arrows follow, again by pseudo-transitivity. By pseudo-transitivity (invoked twice), the lower bullet point is strictly worse than the upper bullet point according to player $b$. This makes the rectangle induced by these bullet points (or equivalently by $s_1$ and $s_4$) a forbidden pattern (PUC).

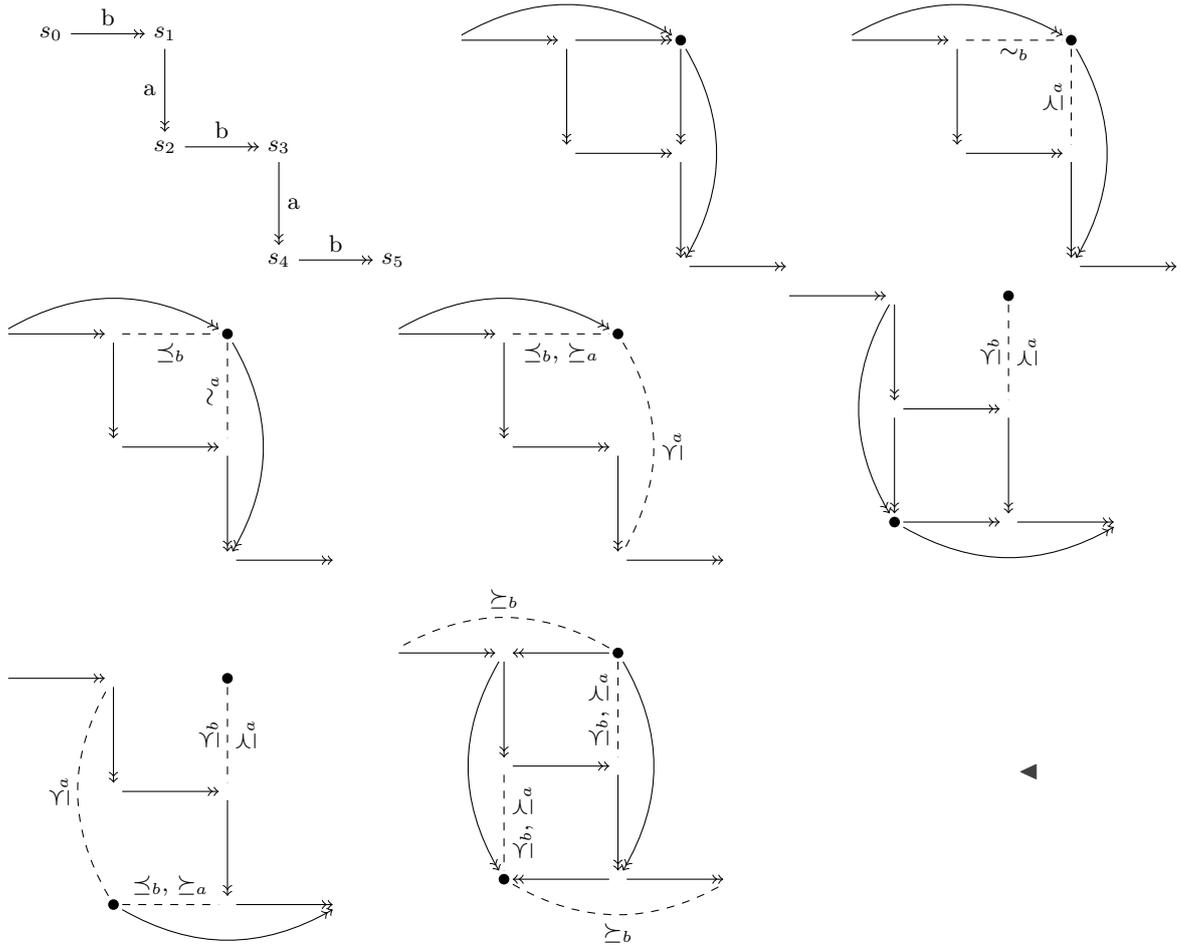
\begin{tabular}{ccc}
\begin{tikzpicture}[node distance=1.5cm]
  \node(x){$s_0$};
  \node(y)[right of = x]{$s_1$};
  \node(z)[below of = y]{$s_2$};
  \node(t)[right of = z]{$s_3$};
  \node(u)[below of = t]{$s_4$};
  \node(v)[right of = u]{$s_5$};
  \node(r)[right of = y]{};
  \node(s)[below of = z]{};

  \draw [->>] (x) to node[above] {b} (y);
  \draw [->>] (y) to node[left] {a} (z);
  \draw [->>] (z) to node[above] {b} (t);
  \draw [->>] (t) to node[right] {a} (u);
  \draw [->>] (u) to node[above] {b} (v);
  \end{tikzpicture}
&
\begin{tikzpicture}[node distance=1.5cm]
  \node(x){};
  \node(y)[right of = x]{};
  \node(z)[below of = y]{};
  \node(t)[right of = z]{};
  \node(u)[below of = t]{};
  \node(v)[right of = u]{};
  \node(r)[right of = y]{};
  \node(s)[below of = z]{};

  \draw [->>] (x) -- (y);
  \draw [->>] (y) -- (z);
  \draw [->>] (z) -- (t);
  \draw [->>] (t) -- (u);
  \draw [->>] (u) -- (v);

  \draw [->>] (y) to (r);
  \draw [->>] (r) to (t);
  \draw [->>] (x) to [bend left=30] (r);
  \draw [->>] (r) to [bend left=30] (u);

  \fill (r) circle (2pt);
  \end{tikzpicture}
&
\begin{tikzpicture}[node distance=1.5cm]
  \node(x){};
  \node(y)[right of = x]{};
  \node(z)[below of = y]{};
  \node(t)[right of = z]{};
  \node(u)[below of = t]{};
  \node(v)[right of = u]{};
  \node(r)[right of = y]{};
  \node(s)[below of = z]{};

  \draw [->>] (x) -- (y);
  \draw [->>] (y) -- (z);
  \draw [->>] (z) -- (t);
  \draw [->>] (t) -- (u);
  \draw [->>] (u) -- (v);

  \draw [dashed] (y) to node[below] {$\sim_b$} (r);
  \draw [dashed] (r) to node[above, rotate=90] {$\succeq_a$} (t);
  \draw [->>] (x) to [bend left=30] (r);
  \draw [->>] (r) to [bend left=30] (u);

  \fill (r) circle (2pt);
  \end{tikzpicture}
\end{tabular}
\begin{tabular}{ccc}
\begin{tikzpicture}[node distance=1.5cm]
  \node(x){};
  \node(y)[right of = x]{};
  \node(z)[below of = y]{};
  \node(t)[right of = z]{};
  \node(u)[below of = t]{};
  \node(v)[right of = u]{};
  \node(r)[right of = y]{};
  \node(s)[below of = z]{};

  \draw [->>] (x) -- (y);
  \draw [->>] (y) -- (z);
  \draw [->>] (z) -- (t);
  \draw [->>] (t) -- (u);
  \draw [->>] (u) -- (v);

  \draw [dashed] (y) to node[below] {$\preceq_b$} (r);
  \draw [dashed] (r) to node[above, rotate=90] {$\sim_a$} (t);
  \draw [->>] (x) to [bend left=30] (r);
  \draw [->>] (r) to [bend left=30] (u);

  \fill (r) circle (2pt);
  \end{tikzpicture}
&
\begin{tikzpicture}[node distance=1.5cm]
  \node(x){};
  \node(y)[right of = x]{};
  \node(z)[below of = y]{};
  \node(t)[right of = z]{};
  \node(u)[below of = t]{};
  \node(v)[right of = u]{};
  \node(r)[right of = y]{};
  \node(s)[below of = z]{};

  \draw [->>] (x) -- (y);
  \draw [->>] (y) -- (z);
  \draw [->>] (z) -- (t);
  \draw [->>] (t) -- (u);
  \draw [->>] (u) -- (v);

  \draw [dashed] (y) to node[below] {$\preceq_b$, $\succeq_a$} (r);
  \draw [->>] (x) to [bend left=30] (r);
  \draw [dashed] (r) to [bend left=30] node[below, rotate=90] {$\preceq_a$} (u);

  \fill (r) circle (2pt);
  \end{tikzpicture}
&
\begin{tikzpicture}[node distance=1.5cm]
  \node(x){};
  \node(y)[right of = x]{};
  \node(z)[below of = y]{};
  \node(t)[right of = z]{};
  \node(u)[below of = t]{};
  \node(v)[right of = u]{};
  \node(r)[right of = y]{};
  \node(s)[below of = z]{};

  \draw [->>] (x) -- (y);
  \draw [->>] (y) -- (z);
  \draw [->>] (z) -- (t);
  \draw [->>] (t) -- (u);
  \draw [->>] (u) -- (v);

  \draw [dashed] (r) to node[above, rotate=90] {$\preceq_b$} node [below, rotate=90] {$\succeq_a$} (t);
  \draw [->>] (s) to (u);
  \draw [->>] (z) to (s);
  \draw [->>] (s) to [bend right=30] (v);
  \draw [->>] (y) to [bend right=30] (s);

  \fill (r) circle (2pt);
  \fill (s) circle (2pt);
  \end{tikzpicture}
\end{tabular}
\begin{tabular}{ccc}
\begin{tikzpicture}[node distance=1.5cm]
  \node(x){};
  \node(y)[right of = x]{};
  \node(z)[below of = y]{};
  \node(t)[right of = z]{};
  \node(u)[below of = t]{};
  \node(v)[right of = u]{};
  \node(r)[right of = y]{};
  \node(s)[below of = z]{};

  \draw [->>] (x) -- (y);
  \draw [->>] (y) -- (z);
  \draw [->>] (z) -- (t);
  \draw [->>] (t) -- (u);
  \draw [->>] (u) -- (v);

  \draw [dashed] (r) to node[above, rotate=90] {$\preceq_b$} node [below, rotate=90] {$\succeq_a$} (t);
  \draw [dashed] (s) to node[above] {$\preceq_b$, $\succeq_a$} (u);
  \draw [->>] (s) to [bend right=30] (v);
  \draw [dashed] (s) to [bend left=30] node[above, rotate=90] {$\preceq_a$} (y);

  \fill (r) circle (2pt);
  \fill (s) circle (2pt);
  \end{tikzpicture}
&
\begin{tikzpicture}[node distance=1.5cm]
  \node(x){};
  \node(y)[right of = x]{};
  \node(z)[below of = y]{};
  \node(t)[right of = z]{};
  \node(u)[below of = t]{};
  \node(v)[right of = u]{};
  \node(r)[right of = y]{};
  \node(s)[below of = z]{};

  \draw [->>] (x) -- (y);
  \draw [->>] (y) -- (z);
  \draw [->>] (z) -- (t);
  \draw [->>] (t) -- (u);
  \draw [->>] (u) -- (v);

  \draw [dashed] (r) to node[above, rotate=90] {$\preceq_b$, $\succeq_a$} (t);
  \draw [->>] (r) to [bend left=30] (u); 
  \draw [dashed] (r) to [bend right=30] node[above] {$\succeq_b$} (x);
  \draw [->>] (r) -- (y);
  \draw [dashed] (z) to node[below, rotate=90] {$\preceq_b$, $\succeq_a$} (s);
  \draw [->>] (y) to [bend right=30] (s);
  \draw [dashed] (s) to [bend right=30] node[below] {$\succeq_b$} (v);
  \draw [->>] (u) -- (s);
 
  \fill (r) circle (2pt);
  \fill (s) circle (2pt);
  \end{tikzpicture}
\end{tabular}
\end{proof}

Theorem~\ref{thm:fp-acycl} below relies on Lemma~\ref{lem:fp-decrease} (but does not refer to sheaths). It implies Corollary~\ref{cor:fp-term} since a strict weak order is a strict order with transitive negation, hence (pseudo-)transitivity.

\begin{theorem}\label{thm:fp-acycl}
If a finite two-player game with transitive and pseudo-transitive preferences does not contain any forbidden pattern, every player performing infinitely many steps in an improvement sequence has a preference with a cycle.
\end{theorem}

\begin{proof}
Let us assume that Player $a$ performs infinitely many steps in some improvement sequence. Since the game is finite, a cycle involving Player $a$ can be extracted from this sequence. Let us consider such a cycle with minimal length: if its length is three or less, it involves Player $a$ only, so she has a preference with a cycle; the length cannot be four either, because there is no forbidden pattern in the game. By minimality of the cycle and transitivity of the preferences, the two players must alternate their individual-improvement steps, so that the length of the cycle is even, so it must be at least six. Lemma~\ref{lem:fp-decrease} then implies that there is, according to Player $a$, an infinite descending sequence of outcomes, so by finiteness of the game there must be a cycle.
\end{proof}

\begin{corollary}\label{cor:fp-term}
If a finite two-player game with strict-weak-order preferences does not contain any forbidden pattern, the collective improvement terminates.
\end{corollary}

The following examples show that each of the three forbidden patterns alone can ruin the FIP, so it makes sense to forbid each of them. The first game below consists of an elementary cycle; the second game has one cycle $(a_1,b_1)\twoheadrightarrow(a_1,b_2)\twoheadrightarrow(a_2,b_2)\twoheadrightarrow(a_2,b_3)\twoheadrightarrow(a_3,b_3)\twoheadrightarrow(a_3,b_1)\twoheadrightarrow(a_1,b_1)$, and only one type of forbidden pattern (PUB) occurring actually six times; the third game has the same cycle, and only one forbidden pattern (PUB) occurring at $\{a_1,a_3\} \times \{b_1,b_3\}$; the fourth game has the same cycle, and only one forbidden pattern (PUC, where the payoffs of both players decrease) occurring at $\{a_1,a_3\} \times \{b_2,b_3\}$; and the fifth game has a longer, staircase-like cycle starting at $(a_1,b_1)$, and only one forbidden pattern (PUC, where only the payoff of Player $b$ decreases) occurring at $\{a_2,a_4\} \times \{b_3,b_4\}$.

\[\begin{array}{ccccc}
\begin{array}{|c|c|}
	  \cline{1-2}
 	    1,0 & 0,1\\
	  \cline{1-2}
	    0,1 & 1,0 \\
	  \cline{1-2}
	 \end{array}
&
\begin{array}{c|c|c|c|}
	\multicolumn{1}{c}{}&
	  \multicolumn{1}{c}{b_1}&
	  \multicolumn{1}{c}{b_2}&
	  \multicolumn{1}{c}{b_3}\\
	  \cline{2-4}
 	  a_1 &  1,0 & 0,1 & 0,0\\
	  \cline{2-4}
	   a_2 & 0,0 & 1,0 & 0,1 \\
	  \cline{2-4}
	   a_3 & 0,1 & 0,0 & 1,0 \\
	  \cline{2-4}
	 \end{array}
&
\begin{array}{c|c|c|c|}
	\multicolumn{1}{c}{}&
	  \multicolumn{1}{c}{b_1}&
	  \multicolumn{1}{c}{b_2}&
	  \multicolumn{1}{c}{b_3}\\
	  \cline{2-4}
  	  a_1 & 2,2 & 0,3 & 2,4\\
	  \cline{2-4}
	  a_2 &  2,2 & 1,1 & 0,2 \\
	  \cline{2-4}
 	  a_3 & 1,1 & 1,1 & 1,0 \\
	  \cline{2-4}
	 \end{array}
\\\\
\begin{array}{c|c|c|c|}
	\multicolumn{1}{c}{}&
	  \multicolumn{1}{c}{b_1}&
	  \multicolumn{1}{c}{b_2}&
	  \multicolumn{1}{c}{b_3}\\
	  \cline{2-4}
 	  a_1 & 2,2 & 0,3 & 2,2\\
	  \cline{2-4}
	  a_2 & 1,1 & 1,1 & 2,2 \\
	  \cline{2-4}
	  a_3 & 1,1 & 1,1 & 3,0 \\
	  \cline{2-4}
	 \end{array}
&
\begin{array}{c|c|c|c|c|}
	\multicolumn{1}{c}{}&
	  \multicolumn{1}{c}{b_1}&
	  \multicolumn{1}{c}{b_2}&
	  \multicolumn{1}{c}{b_3}&
	  \multicolumn{1}{c}{b_4}\\
	  \cline{2-5}
 	  a_1 &  2,0 & 0,1 & 0,0 & 2,0\\
	  \cline{2-5}
	   a_2 & 1,2 & 1,2 & 0,3 & 1,2\\
	  \cline{2-5}
	   a_3 & 1,2 & 1,2 & 1,1 & 2,2 \\
	  \cline{2-5}
	   a_4 & 1,2 & 1,2 & 1,1 & 3,0 \\
	  \cline{2-5}
	 \end{array}
&
\begin{array}{c|c|c|c|}
	\multicolumn{1}{c}{}&
	  \multicolumn{1}{c}{b_1}&
	  \multicolumn{1}{c}{b_2}&
	  \multicolumn{1}{c}{b_3}\\
	  \cline{2-4}
 	  a_1 &  x & y & z\\
	  \cline{2-4}
	   a_2 & z & x & y \\
	  \cline{2-4}
	   a_3 & y & z & x \\
	  \cline{2-4}
	 \end{array}
\end{array}\]

The last example above shows that the strict-weak-order condition in Corollary~\ref{cor:fp-term} is relevant: let $x$, $y$ and $z$ be three outcomes and assume that $y \prec_a x$ and $x \prec_b y$, while both players are indifferent to $z$. Then $\twoheadrightarrow$ has a cycle (similar to the three other $3\times 3$ games above), although it has no forbidden pattern since $z$ occurs in every $2 \times 2$ sub game. 

The two games above that involve a PUC pattern have Nash equilibria: $(a_3,b_2)$ for the fourth game, and $(a_3,b_2)$ and $(a_4,b_2)$ for the fifth game. It is actually always the case when allowing the weakly forbidden pattern PUC as shown by Theorem~\ref{thm:sfp-ne} after Lemma~\ref{lem:fp-max-impr} below, which is a useful variant of Lemma~\ref{lem:fp-decrease}. This new lemma refers to the \emph{maximising improvement}, which is the restriction of  $\twoheadrightarrow$ that forces the improving player to maximise the outcome.

\begin{lemma}\label{lem:fp-max-impr}
Let be a two-player game with strict-weak-order preferences. Let $s_0\twoheadrightarrow_b s_1\twoheadrightarrow_a s_2 \twoheadrightarrow_b s_3\twoheadrightarrow_a s_4$ be a path of maximising improvement, and let us assume that there is, within the sheath of the path, no shorter maximising-improvement path from $s_1$ to $s_4$. Also assume that no strongly forbidden pattern (EC or PUB) occurs within the sheath. Then $v(s_3)\prec_b v(s_1)$ or $(s_1^a,s_3^b)$ is an NE.
\end{lemma}

\begin{proof}
Such a path is represented in the first picture below. Thanks to Observation~\ref{obs:comp-pat} and the assumption that there is no strongly forbidden pattern within the sheath, the rectangle that is determined by the profiles $s_1$, $s_2$, and $s_3$ may match six different patterns. Since the preferences are strict weak orders, the first three patterns (PHC and PUC together, IS, and IA) contradict minimality of the path, as shown by the additional improvement arrows in the second, third, and fourth pictures below. If it matches the CS pattern, as in the fifth picture below, the longer improvement arrow starting at $s_0$ holds by transitivity; the longer dashed inequality holds not to contradict minimality of the path; so the black spot $(s_1^a,s_3^b)$ is an NE by maximality of the improvement. For the CA pattern in the sixth picture, the short additional arrow starting at the black spot is here due to the minimality of the path, and $v(s_3)\prec_b v(s_1)$ holds by transitivity.

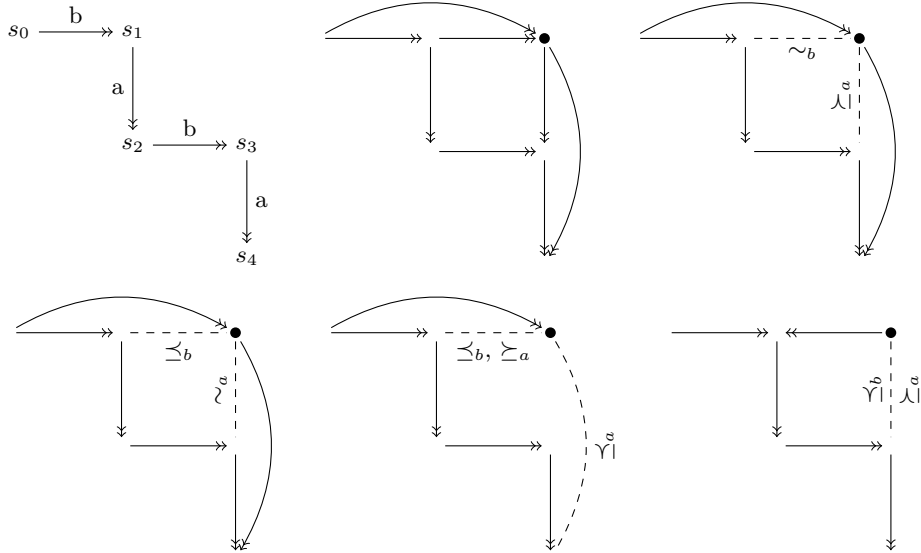
\begin{tabular}{ccc}
\begin{tikzpicture}[node distance=1.5cm]
  \node(x){$s_0$};
  \node(y)[right of = x]{$s_1$};
  \node(z)[below of = y]{$s_2$};
  \node(t)[right of = z]{$s_3$};
  \node(u)[below of = t]{$s_4$};
  \node(r)[right of = y]{};

  \draw [->>] (x) to node[above] {b} (y);
  \draw [->>] (y) to node[left] {a} (z);
  \draw [->>] (z) to node[above] {b} (t);
  \draw [->>] (t) to node[right] {a} (u);
  \end{tikzpicture}
&
\begin{tikzpicture}[node distance=1.5cm]
  \node(x){};
  \node(y)[right of = x]{};
  \node(z)[below of = y]{};
  \node(t)[right of = z]{};
  \node(u)[below of = t]{};
  \node(r)[right of = y]{};

  \draw [->>] (x) -- (y);
  \draw [->>] (y) -- (z);
  \draw [->>] (z) -- (t);
  \draw [->>] (t) -- (u);

  \draw [->>] (y) to (r);
  \draw [->>] (r) to (t);
  \draw [->>] (x) to [bend left=30] (r);
  \draw [->>] (r) to [bend left=30] (u);

  \fill (r) circle (2pt);
  \end{tikzpicture}
&
\begin{tikzpicture}[node distance=1.5cm]
  \node(x){};
  \node(y)[right of = x]{};
  \node(z)[below of = y]{};
  \node(t)[right of = z]{};
  \node(u)[below of = t]{};
  \node(r)[right of = y]{};

  \draw [->>] (x) -- (y);
  \draw [->>] (y) -- (z);
  \draw [->>] (z) -- (t);
  \draw [->>] (t) -- (u);

  \draw [dashed] (y) to node[below] {$\sim_b$} (r);
  \draw [dashed] (r) to node[above, rotate=90] {$\succeq_a$} (t);
  \draw [->>] (x) to [bend left=30] (r);
  \draw [->>] (r) to [bend left=30] (u);

  \fill (r) circle (2pt);
  \end{tikzpicture}
\end{tabular}

\begin{tabular}{ccc}
\begin{tikzpicture}[node distance=1.5cm]
  \node(x){};
  \node(y)[right of = x]{};
  \node(z)[below of = y]{};
  \node(t)[right of = z]{};
  \node(u)[below of = t]{};
  \node(r)[right of = y]{};

  \draw [->>] (x) -- (y);
  \draw [->>] (y) -- (z);
  \draw [->>] (z) -- (t);
  \draw [->>] (t) -- (u);

  \draw [dashed] (y) to node[below] {$\preceq_b$} (r);
  \draw [dashed] (r) to node[above, rotate=90] {$\sim_a$} (t);
  \draw [->>] (x) to [bend left=30] (r);
  \draw [->>] (r) to [bend left=30] (u);

  \fill (r) circle (2pt);
  \end{tikzpicture}
&
\begin{tikzpicture}[node distance=1.5cm]
  \node(x){};
  \node(y)[right of = x]{};
  \node(z)[below of = y]{};
  \node(t)[right of = z]{};
  \node(u)[below of = t]{};
  \node(r)[right of = y]{};

  \draw [->>] (x) -- (y);
  \draw [->>] (y) -- (z);
  \draw [->>] (z) -- (t);
  \draw [->>] (t) -- (u);

  \draw [dashed] (y) to node[below] {$\preceq_b$, $\succeq_a$} (r);
  \draw [->>] (x) to [bend left=30] (r);
  \draw [dashed] (r) to [bend left=30] node[below, rotate=90] {$\preceq_a$} (u);

  \fill (r) circle (2pt);
  \end{tikzpicture}
&
\begin{tikzpicture}[node distance=1.5cm]
  \node(x){};
  \node(y)[right of = x]{};
  \node(z)[below of = y]{};
  \node(t)[right of = z]{};
  \node(u)[below of = t]{};
  \node(r)[right of = y]{};

  \draw [->>] (x) -- (y);
  \draw [->>] (y) -- (z);
  \draw [->>] (z) -- (t);
  \draw [->>] (t) -- (u);
  \draw [->>] (r) -- (y);

  \draw [dashed] (r) to node[above, rotate=90] {$\preceq_b$} node [below, rotate=90] {$\succeq_a$} (t);

  \fill (r) circle (2pt);
  \end{tikzpicture}
\end{tabular}

\end{proof}

Theorem~\ref{thm:sfp-ne} below is a straightforward corollary of Lemma~\ref{lem:fp-max-impr}, and still refers to sheaths.  

\begin{theorem}\label{thm:sfp-ne}
If a finite two-player game with strict-weak-order preferences has no strongly forbidden pattern, the maximising improvement induces a sheath containing an NE.  
\end{theorem}

\begin{proof}
Let us assume that it does not, so by finiteness of the game there is a cycle of maximising improvement with no NE in its sheath. Let us consider such a cycle of minimal length, say, $s_0\twoheadrightarrow_b s_1\twoheadrightarrow_a s_2\twoheadrightarrow_b \dots \twoheadrightarrow_b s_{2n-1}\twoheadrightarrow_a s_0$. By Lemma~\ref{lem:fp-max-impr} the cycle $v(s_1) \succ_b v(s_3) \succ_b \dots \succ_b v(s_{2n-1})\succ_b v(s_1)$ holds, contradiction.
\end{proof}

The conditions of Theorem~\ref{thm:sfp-ne} are not necessary, though, as shown below. The leftmost game is a PUB-pattern and the second game has an elementary cycle but maximising improvement terminates. By the way, the third game below has an NE although the maximising improvement may generate a sheath without NE. More importantly, the first game below is an exact potential game, and the fourth game is not even an ordinal potential game, whereas it is not a forbidden pattern. It shows that the class of (ordinal) potential games and the class of games without forbidden patterns are not included one in another.  
 
\begin{tabular}{c@{\hspace{1cm}}c@{\hspace{1cm}}c@{\hspace{1cm}}c}
$\begin{array}{|c|c|}
	  \cline{1-2}
 	    0,0 & 1,3 \\
	  \cline{1-2}
	  1,0 & 0,1 \\
	  \cline{1-2}
	 \end{array}$
&	
$\begin{array}{|c|c|c|}
	  \cline{1-3}
 	   1,0 & 0,1 & 0,2\\
	  \cline{1-3}
	    0,1 & 1,0 & 0,0 \\
	  \cline{1-3}
	 \end{array}$
&
$\begin{array}{|c|c|c|}
	  \cline{1-3}
 	    1,0 & 0,1 & 0,0\\
	  \cline{1-3}
	    0,1 & 1,0 & 0,0 \\
	  \cline{1-3}
	    0,0 & 0,0 & 1,1 \\
	  \cline{1-3}
	 \end{array}$
&
$\begin{array}{|c|c|}
	  \cline{1-2}
 	    0,0 & 0,0 \\
	  \cline{1-2}
	  1,0 & 0,1 \\
	  \cline{1-2}
	 \end{array}$
\end{tabular}

Theorem~\ref{thm:pot-gs} below states that the game structures where assigning arbitrary acyclic preferences always yields the FIP are exactly those without forbidden rectangles. At first it might have seemed difficult (\textit{e.g.} NP-hard) to decide which structure has this property, but this characterisation shows that it is at most quadratic in the number of profiles.

\begin{theorem}\label{thm:pot-gs}
Given a two-player game structure, the following are equivalent.
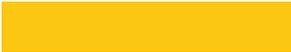
\begin{enumerate}
\item\label{thm:pot-gs1} The game structure has no forbidden rectangle. More specifically, for all $s^1_a$ and $s^2_a$ (resp. $s^1_b$ and $s^2_b$) strategies for player $a$ (resp. $b$), either $v(s^1_a,s^1_b) = v(s^2_a,s^1_b)$, or $v(s^2_a,s^1_b) = v(s^2_a,s^2_b)$, or $v(s^2_a,s^2_b) = v(s^1_a,s^2_b)$, or $v(s^1_a,s^2_b) = v(s^1_a,s^1_b)$.
\item\label{thm:pot-gs2} Equipping the game structure with any two acyclic preferences (resp. instantiating it with win-lose outcomes) yields a game without $\twoheadrightarrow$-cycles of length four.
\item\label{thm:pot-gs3} Equipping the game structure with any two acyclic preferences yields a game without $\twoheadrightarrow$-cycles.
\item\label{thm:pot-gs4} Equipping the game structure with any two acyclic preferences yields a game whose subgames all have NE.
\end{enumerate}
\end{theorem}

\begin{proof}
\ref{thm:pot-gs1}$\Leftrightarrow$\ref{thm:pot-gs2} and \ref{thm:pot-gs3}$\Rightarrow$\ref{thm:pot-gs4} and \ref{thm:pot-gs4}$\Rightarrow$\ref{thm:pot-gs1} are straightforward. For \ref{thm:pot-gs1}$\Rightarrow$\ref{thm:pot-gs3}, let us augment the preferences by transitive closures followed by linear extension, and let us invoke Corollary~\ref{cor:fp-term} since absence of forbidden rectangle in the structure implies absence of forbidden pattern in the games. 
\end{proof}

\subsection{The three-player case}\label{sect:afi-3p}

It is unclear whether and how Theorem~\ref{thm:pot-gs} may be generalised to three-player game structures, but some facts are worth noting before embarking on such generalisation. First, the leftmost two tables below represent a game structure with players $a$, $b$ and $c$, where $c$ chooses the left or right array. Slicing the game, \textit{i.e.}, fixing a strategy for one of the players always yields a subgame without $\twoheadrightarrow$-cycles, whatever the acyclic preferences of the remaining two players may be, but defining $z<_ay<_ax$ and $y<_bt<_bz$ and $x<_ct<_cy$ yields the cycle $(1,1,1)\twoheadrightarrow_c (1,1,2)\twoheadrightarrow_b (1,2,2)\twoheadrightarrow_a (2,2,2)\twoheadrightarrow_b (2,1,2)\twoheadrightarrow_c(2,1,1)\twoheadrightarrow_a (1,1,1)$. This means that ruling out two-player cycles is not sufficient to rule out three-player cycles.

\[\begin{array}{cc@{\hspace{2cm}}cc}
\begin{array}{c|c|c|}  
  	\multicolumn{1}{c}{}&
	  \multicolumn{1}{c}{b_1}&
	  \multicolumn{1}{c}{b_2}\\
	  \cline{2-3}
 	a_1 & x & z \\
	  \cline{2-3}
	  a_2 &  y & y \\
	  \cline{2-3}
	 \end{array}
&
\begin{array}{c|c|c|}
 	\multicolumn{1}{c}{}&
	  \multicolumn{1}{c}{b_1}&
	  \multicolumn{1}{c}{b_2}\\
	  \cline{2-3}
 	 a_1 &  t & z \\
	  \cline{2-3}
	    a_2 &   t & y \\
	  \cline{2-3}
	 \end{array}
&
\begin{array}{c|c|c|c|}
	\multicolumn{1}{c}{}&
	  \multicolumn{1}{c}{b_1}&
	  \multicolumn{1}{c}{b_2}&
	  \multicolumn{1}{c}{b_3}\\
	  \cline{2-4}
 	  a_1 &  x & x & x\\
	  \cline{2-4}
	   a_2 & x & y & y \\
	  \cline{2-4}
	 \end{array}
&
\begin{array}{c|c|c|c|}
	\multicolumn{1}{c}{}&
	  \multicolumn{1}{c}{b_1}&
	  \multicolumn{1}{c}{b_2}&
	  \multicolumn{1}{c}{b_3}\\
	  \cline{2-4}
 	  a_1 &  z & y & z\\
	  \cline{2-4}
	   a_2 & z & y & y \\
	  \cline{2-4}
	 \end{array}
\end{array}\]

\noindent Second, although the absence of cycles in a two-player game structure is equivalent to the absence of a cycle in all of its $2\times 2$ subgame structures, the right-hand three-player game structure above enjoys different properties: equipping it with acyclic preferences cannot yield cycles in any of its three $2\times 2\times 2$ subgame structures, but defining $z<_ay<_ax$ and $x<_by<_bz$ and $z<_cx<_cy$ yields the cycle $(2,1,1)\twoheadrightarrow_b (2,2,1)\twoheadrightarrow_a (1,2,1)\twoheadrightarrow_c(1,2,2)\twoheadrightarrow_b(1,3,2)\twoheadrightarrow_a(2,3,2)\twoheadrightarrow_b(2,1,2)\twoheadrightarrow_c(2,1,1)$

\subsection{$\epsilon$-improvement and $\epsilon$-Nash equilibrium in infinite games}\label{sect:ei-ene-ig}

Definition~\ref{def:epsilon-NE} below recalls the well-known notion of $\epsilon$-Nash equilibrium. Lemma~\ref{lem:strong-uniform-continuity} relates to yet differs from uniform continuity and its proof is in appendix. Finally, Proposition~\ref{prop:fp-ig} generalises Corollary~\ref{cor:fp-term} for infinite games.

\begin{definition}[$\epsilon$-improvement and $\epsilon$-Nash equilibrium]\label{def:epsilon-NE}
In a two-player game where the outcomes are real-valued payoff functions, an $\epsilon$-improvement is an improvement w.r.t. the following $\epsilon$-preferences, and an $\epsilon$-NE is an NE w.r.t. to these preferences.
\[(x,y)<_a^\epsilon (x',y')\,:=\, x+\epsilon < x'\qquad (x,y)<_b^\epsilon (x',y')\,:=\, y+\epsilon < y'\]
\end{definition}

\begin{lemma}\label{lem:strong-uniform-continuity}
Let $f:S \to \mathbb{R}^n$ be a continuous function on a compact space $S :=\prod_{1\leq i\leq n}S_i$, where $1 \leq n$, and let $0 < \epsilon$. There exist partitions $(S_{ij})_{1\leq j \leq m}$ of the $S_i$ such that for all $j\in \{1,\dots,m\}^n$, for all $x,y\in \prod_{1\leq i\leq n}S_{i,j(i)}$, we have $|f(x)-f(y)| < \epsilon$.
\end{lemma}

\begin{proposition}\label{prop:fp-ig}
Let a game $\langle \{a,b\},S_a,S_b,\mathbb{R}^2,v,<_a^0,<_b^0\rangle$ have compact strategy sets, a continuous $v$ for the product topology, and no forbidden pattern. For $0 < \epsilon$, the collective $\epsilon$-improvement terminates on an $\epsilon$-NE.
\end{proposition}

\begin{proof}
By Lemma~\ref{lem:strong-uniform-continuity}, let $(S_{aj})_{1\leq j \leq m}$ and $(S_{bj})_{1\leq j \leq m}$ be partitions of $S_a$ and $S_b$, respectively, such that for all $j,k\in \{1,\dots,m\}$, for all $s,s'\in S_{aj} \times S_{bk}$, we have $|v(s)-v(s')| < \epsilon / 3$. Let us derive a game $\langle \{a,b\},\{S_{a1},\dots, S_{am}\},\{S_{b1},\dots, S_{bm}\},\{S_{a1},\dots, S_{am}\} \times \{S_{b1},\dots, S_{bm}\},\mathrm{id},\prec_a,\prec_b\rangle$, where $(S_{aj},S_{bk}) \prec_a (S_{aj'},S_{bk'})$ if there exist $s$ and $s'$ in $S_{aj} \times S_{bk}$ and $S_{aj'} \times S_{bk'}$, respectively, such that $v(s)  <_a^\epsilon v(s')$. (And likewise for $b$.) So, each $\epsilon$-improvement step in the original game is matched by an improvement step in the derived game. Also, note that $(S_{aj},S_{bk}) \prec_a (S_{aj'},S_{bk'})$ implies $v(s) + \epsilon /3 < v(s')$ for all $s$ and $s'$ in $S_{aj} \times S_{bk}$ and $S_{aj'} \times S_{bk'}$, respectively. So, neither the derived game has forbidden patterns, so improvement therein terminates by Corollary~\ref{cor:fp-term}, and $\epsilon$-improvement terminates in the original game.
\end{proof}

\section{Inductive description of the matrices without forbidden rectangles}\label{sect:id-mwfs}

Section~\ref{sect:cife} showed the importance of the game structures, \textit{i.e.}, matrices without forbidden rectangles. The current section gives an analytical, inductive description of these, which is then used to prove a further result about them. Let us start by a useful observation below.

\begin{observation}\label{obs:nfs-perm}
Whether a matrix has forbidden rectangles or not is preserved under permutation of rows and columns.
\end{observation}

Next, let us consider a matrix where the entry $x$ occurs. Lemma~\ref{lem:ind-fs} below (proof is in the appendix) states that if the matrix has no forbidden rectangle, it looks like the one below up to row and column permutation, where $x$ occurs exactly in the area above the staircase on the top-left corner, where the $y_i$-stripes and the $z_j$-stripes correspond to area with constant entry, $y_i$ or $z_j$, where the big white rectangle has no forbidden rectangle, and finally where every entry at the intersection of any $y_i$-row and any $z_j$-column has value $y_i$ or $z_j$.

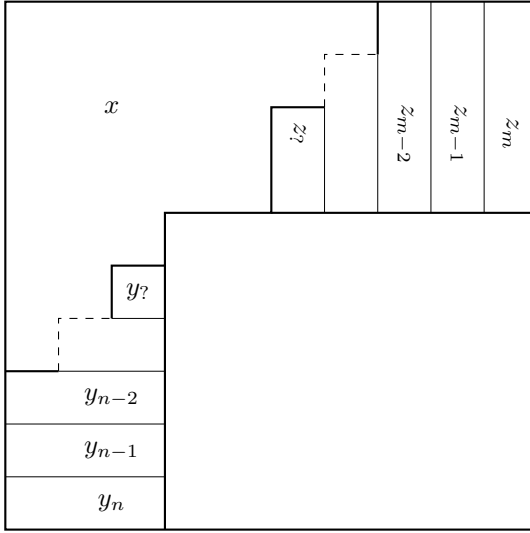
\begin{tikzpicture}[scale=.7]
\draw[thick] (0,0) rectangle (10,10);

\draw[thick] (0,3) -- (1,3);
\draw[dashed] (1,3) -- (1,4) -- (2,4);
\draw[thick] (2,4) -- (2,5) -- (3,5) -- (3,6) -- (5,6) -- (5,8) -- (6,8);
\draw[dashed] (6,8) -- (6,9) -- (7,9);
\draw[thick] (7,9) -- (7,10);

\node at (2,8) {$x$};
\node at (2,.5) {$y_n$};
\node at (2,1.5) {$y_{n-1}$};
\node at (2,2.5) {$y_{n-2}$};
\node at (2.5,4.5) {$y_?$};

\node at (5.5,7.5) [rotate = -90] {$z_?$};
\node at (7.5,7.5) [rotate = -90] {$z_{m-2}$};
\node at (8.5,7.5) [rotate = -90] {$z_{m-1}$};
\node at (9.5,7.5) [rotate = -90] {$z_m$};

\draw[thick] (3,5) -- (3,0);
\draw[very thin] (6,8) -- (6,6);
\draw[very thin] (7,9) -- (7,6);
\draw[very thin] (8,10) -- (8,6);
\draw[very thin] (9,10) -- (9,6);

\draw[very thin] (0,1) -- (3,1);
\draw[very thin] (0,2) -- (3,2);
\draw[very thin] (1,3) -- (3,3);
\draw[very thin] (2,4) -- (3,4);
\draw[thick] (5,6) -- (10,6);
\end{tikzpicture}

\begin{lemma}\label{lem:ind-fs}
Let $A=(a_{ij})_{1\leq i\leq n\,\wedge 1\leq j\leq m}$ be a matrix over some set $X$ and let $x\in X$ occurs in $A$. If $A$ has no forbidden rectangle, the following assertions hold: 

\begin{itemize}
\item There exist $\varphi$ and $\theta$ permutations of $\{1,\dots,n\}$ and $\{1,\dots,m\}$, respectively, and a non-increasing function $s:\{1,\dots,m\}\to\{0,\dots,n\}$ such that $1\leq s(1)$ and $i\leq s(j)\,\Leftrightarrow\, b_{ij}=x$, where $b_{ij}:=a_{\varphi(i)\theta(j)}$.
\item There exists also a natural number $k\leq m$ such that:
\begin{enumerate}
\item\label{ind-fs:cond1} if $k<m$, then $s(k+1)<s(k)$, 
\item\label{ind-fs:cond2} if $k<m$ and $s(j)<i\leq i'\leq s(k+1)$, then $b_{ij}=b_{i'j}$,
\item\label{ind-fs:cond3}if $j\leq j'\leq k$ and $s(j)<i$, then $b_{ij}=b_{ij'}$.
\item\label{ind-fs:cond4} if $k<m$, the submatrix $(b_{ij})_{s(k+1)+1\leq i\leq n\,\wedge k+1\leq j\leq m}$ has no forbidden rectangle.
\item\label{ind-fs:cond5} if $k<m$ and $s(k)<i$ and $s(j)<s(k+1)$, then $b_{ij}=b_{ik}$ or $b_{ij}=b_{s(k+1),j}$.
\end{enumerate}
\end{itemize}
\end{lemma}
The implication stated in Lemma~\ref{lem:ind-fs} is not an equivalence: indeed, the left-hand matrix below satisfies all the assertions from Lemma~\ref{lem:ind-fs} but has a forbidden rectangle in its lower-left corner. Furthermore, however the middle partial matrix below may be completed with entries different from $x$, it has always a forbidden rectangle. (To prevent it, the entry at the lower-left corner should be equal to $z$ and $t$ at once.)
\[\begin{array}{c@{\hspace{2cm}}c@{\hspace{2cm}}c}
\begin{array}{|c@{\;\vline\;}c@{\;\vline\;}c|}
	\cline{1-3}
	 x & x & t\\
	\cline{1-3}
	 x & y & y\\
	\cline{1-3}
	 t & z & t\\
	\cline{1-3}
\end{array}
&
\begin{array}{|c@{\;\vline\;}c@{\;\vline\;}c|}
	\cline{1-3}
	 x & x & \\
	\cline{1-3}
	 x & y & y\\
	\cline{1-3}
	  & z & t\\
	\cline{1-3}
\end{array}
&
\begin{array}{|c@{\;\vline\;}c|}
	\cline{1-2}
	 x & y\\
	\cline{1-2}
	 y & x\\
	\cline{1-2}
\end{array}
\end{array}\]

Although Lemma~\ref{lem:ind-fs} is not an equivalence, its inductive flavour is invoked in Lemma~\ref{lem:or} below, whose proof is in the appendix, to prove a simple necessary condition for a matrix to have no forbidden rectangle, which then implies the straightforward Corollary~\ref{cor:num-fs}.

\begin{lemma}\label{lem:or}
Let $A=(a_{ij})_{1\leq i\leq n\,\wedge 1\leq j\leq m}$ be a matrix over some set $X$. If $A$ has no forbidden rectangle, there exist $(y_1,\dots,y_n)\in X^n$ and $(z_1,\dots,z_m)\in X^m$ such that $a_{ij}=y_i$ or $a_{ij}=z_j$ whenever $1\leq i\leq n$ and $1\leq j\leq m$.
\end{lemma}
\begin{corollary}\label{cor:num-fs}
$n\times m$ matrices without forbidden rectangles have at most $n+m$ entries.
\end{corollary}

Note that the necessary condition from Lemma~\ref{lem:or} is not sufficient. Indeed the property holds for the forbidden rectangle above (on the right-hand side) with vectors $(x,x)$ and $(y,y)$.

Definition~\ref{def:scd} below gives two inductive constructors (or decompositions) of games. Then Theorem~\ref{thm:fr-d}, proved in the appendix, shows that these constructors suffice to build all the matrices without forbidden rectangles. Note that it also implies Corollary~\ref{cor:num-fs} easily. Theorem~\ref{thm:fr-d} will probably turn out to be a more accurate inductive description than Lemma~\ref{lem:ind-fs}.
\begin{definition}[Stripe and corner decomposition]\label{def:scd}
If a two-player game structure $G$ can be decomposed, up to permutations of players, rows and columns, into the left-hand (resp. right-hand) structure below, it is said to be stripe-decomposable (resp. corner-decomposable), where $x$, $y$ and $z$ are pairwise distinct outcomes, and where $x/z$ means either $x$ or $z$.
\[\begin{array}{c@{\hspace{2cm}}c}
\begin{array}{|c|c|c|c|}
	  \cline{1-4}
	  x &  x & \dots & x\\
	  \cline{1-4}
	  \multicolumn{4}{|c|}{}\\
	  \multicolumn{4}{|c|}{G'}\\
	  \multicolumn{4}{|c|}{}\\
	  \cline{1-4}
	 \end{array}
&
\begin{array}{|c|c|c|c|c|c|}
	  \cline{1-6}
		x & x & z & x/z & \dots & x/z\\
	  \cline{1-6}
		x & y & y & x/y & \dots & x/y\\
	  \cline{1-6}
		z & y & z & y/z & \dots & y/z\\
	  \cline{1-6}
		x/z & x/y & y/z & \multicolumn{3}{c|}{}\\
	  \cline{1-3}
  	 \vdots & \vdots & \vdots & \multicolumn{3}{c|}{G'}\\
	  \cline{1-3}
		x/z & x/y & y/z & \multicolumn{3}{c|}{}\\
	  \cline{1-6}
	 \end{array}
\end{array}\]
\end{definition}

\begin{theorem}\label{thm:fr-d}
Every game structure without forbidden rectangles is stripe-decomposable or corner-decomposable.
\end{theorem}

\section{Weakly-terminating (maximising) improvement}\label{sect:wtmi}

To every finite game one can associate a Markov (resp. maximising Markov) chain with strategy profiles as states, and where positive transition probabilities correspond exactly to improvements (resp. maximising improvements) or self-loops. Observation~\ref{obs:asc-mc} below means that a game is weakly acyclic iff some/any associated Markov chain converges towards NE almost surely. Note that such processes are memoryless and therefore simple to implement, which justifies the notion of weak acyclity. Observation~\ref{obs:asc-mc} follows from basic probability theory and therefore must be known already. 

\begin{observation}\label{obs:asc-mc}
Given a game with acyclic preferences, the assertions below are equivalent.
\begin{enumerate}
\item Improvement (resp. maximising improvement) in the game is weakly terminating.
\item The NE are exactly the profiles with (possibly) positive stationary measure for some/any associated Markov chain (resp. maximising Markov chain)
\end{enumerate}
\end{observation}

It is also easy to prove weak acyclicity in two-player antagonist (\textit{e.g.} zero-sum) games.

\begin{observation}\label{obs:zs-wa}
If a two-player game with antagonist preferences has an NE, it is weakly acyclic in at most three steps.
\begin{proof}
Let $y$ be the value of the game. If the starting profile involves one optimal strategy of one player and one non-optimal strategy of the opponent, letting the opponent choose an optimal strategy yields an NE. Let us now assume that the starting profile involves no optimal strategy. If it yields an outcome that is worse than $y$ for one player, letting this player choose an optimal strategy reduces to the previous case; if it yields $y$, letting one player make an improvement step reduces to the now-previous case.  
\end{proof}
\end{observation}

Whereas the FIP is closed under taking subgames, it is easy to see that weak acyclicity is not. On the other hand, single-profile subgames are of course weakly acyclic under irreflexive preferences, so the next interesting question is whether all non-trivial weakly acyclic games (with acyclic preferences) have a maximal proper subgame that is also weakly acyclic. The answer is no, as shown by the left-hand game below, where empty cells represent the payoffs $(0,0)$. The game is weakly acyclic with one unique NE at the bottom-right corner. However, deleting one strategy of either of the players yields a game that is not weakly acyclic.

\begin{tabular}{c@{\hspace{1cm}}c}
$\begin{array}{|c@{\;\vline\;}c@{\;\vline\;}c@{\;\vline\;}c@{\;\vline\;}c@{\;\vline\;}c@{\;\vline\;}c@{\;\vline\;}c@{\;\vline\;}c@{\;\vline\;}c|}
	\cline{1-10}
	1,3 & 3,1 &&&&&& 5,1 &&\\
	\cline{1-10}
	2,1 & 1,3 & 3,1 &&&5,1&&&&\\
	\cline{1-10}
	& 2,1 & 1,3 & 3,1 &&6,4&4,6&&&\\
	\cline{1-10}
	&&2,1 & 1,3 & 3,1 &4,6&6,4&&&\\
	\cline{1-10}
	&&&2,1 & 1,3 & 3,1 && 6,4& 4,6&\\
	\cline{1-10}
	&&&& 2,1 & 1,3 & 3,1 &4,6&6,4&\\
	\cline{1-10}
	&&&&& 2,1 & 1,3 & 3,1 &&\\
	\cline{1-10}
	&&&&&& 2,1 & 1,3 & 3,1 &\\
	\cline{1-10}
	&&&&&&& 2,1 & 1,3 & 0,3\\
	\cline{1-10}
\end{array}$
&
$\begin{array}{|c@{\;\vline\;}c@{\;\vline\;}c@{\;\vline\;}c|}
	\cline{1-4}
	x & y & z & z  \\
	\cline{1-4}
	y & x & z & z  \\
	\cline{1-4}
	z & z & x & z \\
	\cline{1-4}	
	z & z & z & z \\
	\cline{1-4}	
\end{array}$
\end{tabular}

But the answer is yes when considering the maximising improvement in two-player games, as shown by Proposition~\ref{prop:wa-mi-sg} below. This provides a means to use induction on the number of strategies to prove properties of games with weakly terminating maximising improvement.

\begin{proposition}\label{prop:wa-mi-sg}
Let $g = \langle (S_1,S_2),O,v,(<_1,<_2)\rangle$ be a finite two-player game where $<_1$ and $<_2$ are linear orders and maximising improvement is weakly terminating. If $1<|S_1|,|S_2|$, there exists a strategy $s_i\in S_i$ such that maximising improvement is also weakly terminating in $g' := \langle (S'_1,S'_2),O,v',(<_1,<_2)\rangle$, where $S'_i := S_i\backslash\{s_i\}$ and $S'_{3-i} := S_{3-i}$ and $v' := v|_{S'_1\times S'_2}$.
\begin{proof}
Let us make a case disjunction. First case, let us assume that every profile of $g$ either is an NE or leads to an NE in one (maximising) improvement step, and let us make a nested case disjunction. First sub-case, some strategy $s_i\in S_i$ of some player $i$ is involved in no NE of $g$, so $s_i$ witnesses the claim since every profile that does not involve $s_i$ is either an NE (of $g$ and therefore of $g'$) or leads in one improvement step to an NE that does not involve $s_i$. Second sub-case, every strategy is involved in an NE. Let $s_i \in S_i$ with $i = 1$ and let $s'$ be a profile of $g'$, \textit{i.e.} such that $s'_1 \neq s_1$. If player $1$ can do a maximising improvement step from $s'$ in $g'$, let her do it and let us rename the new profile $s'$. If player $2$ cannot improve from $s'$, it is an NE in $g'$; otherwise recall that by assumption there exists $s_2 \in S_2$ such that $(s'_1,s_2)$ is an NE in $g$. Since $v(s') <_2 v(s'_1,s_2)$, this witnesses an improvement step from $s'$ to an NE.

Second case, let $s$ be a profile such that the shortest maximising improvement paths from $s$ to an NE are as long as possible in $g$, and let $s\to t\to\dots\to u$ be such a path. Note that the minimality condition ensures that the sequence involves each strategy at most once, possibly besides at the end points $s$ and $u$. Let us assume that, say, $t_2 = s_2$ and prove that the maximising improvement is weakly terminating in $g' = \langle (S_1,S_2\backslash\{s_2\}),O,v,(<_1,<_2)\rangle$. Let $t^1$ be a profile in $g'$ and let $t^1\to\dots\to t^n$ be a shortest maximising improvement path from $t^1$ to an NE in $g$. Let us make a three-fold nested case disjunction: first sub-case, $t^k_2 \neq s_2$ for all $k$, so $t^1\to\dots\to t^n$ is also a shortest maximising improvement path from $t^1$ to an NE in $g'$; second sub-case, $t^k_2 = s_2$ for some $1 < k < n$. If $s\to t\to\dots\to u$ is longer than $t^k\to\dots\to t^n$, the maximising improvement path $s\to t^{k+1}\to\dots\to t^n$ contradicts the minimality of $s\to t\to\dots\to u$, and if $s\to t\to\dots\to u$ is not longer than $t^k\to\dots\to t^n$, then $t^1\to\dots\to t^n$ contradicts the maximality of $s\to t\to\dots\to u$; third sub-case, $t^n_2 = s_2$, so $s\to t^n$ is a path from $s$ to an NE, which brings us back to the first case of the main case disjunction.
\end{proof}
\end{proposition}

Instead of mere games, let us now consider a structure that always yields a weakly acyclic game when instantiated with payoffs $(2,0)$, $(1,1)$, and $(0,2)$. Along the lines of \cite{SLR14}, another interesting question is whether weak acyclicity always holds when the structure is instantiated with arbitrary acyclic preferences. It holds for all antagonist preferences by Observation~\ref{obs:zs-wa} and the transfer theorem from \cite{SLR14}, but it fails for some acyclic preferences, as shown by the right-hand structure above. All antagonist preferences yield a weakly acyclic game, but if $z <_a y <_a x$ and $z <_b x <_b y$, improvement get stuck in a cycle in the upper-left corner.

\subparagraph*{Acknowledgements}

I thank Dietmar Berwanger, Yvan Le Borgne, and Victor Poupet for useful discussions and comments.

\section{Future work}

This article suggests a few natural directions for future work. First, most of the results obtained here involve two-player games only, so it would be interesting to generalise these for multi-player games. Second, the forbidden patterns led to a sufficient condition ensuring the FIP, in Corollary~\ref{cor:fp-term}, which was almost necessary but still had a gap. It would be interesting to see whether considering patterns bigger than $2 \times 2$ subgames could enable a meaningful refinement of Corollary~\ref{cor:fp-term} and Theorem~\ref{thm:sfp-ne}. Third, similar to Theorem~\ref{thm:pot-gs}, it would be interesting to characterise, or at least approximate precisely, the (even two-player) game structures that always yield weakly acyclic games however they may be instantiated with reasonable preferences.

\bibliography{article}

\newpage

\appendix

\begin{proof}[Proof of Observation~\ref{obs:comp-pat}]
The following tree describes the nested case disjunctions. In each diagram Players $a$ and $b$ can move vertically and horizontally, respectively. Unnecessary or retrievable information is not displayed, which hopefully makes the diagrams easier to read.

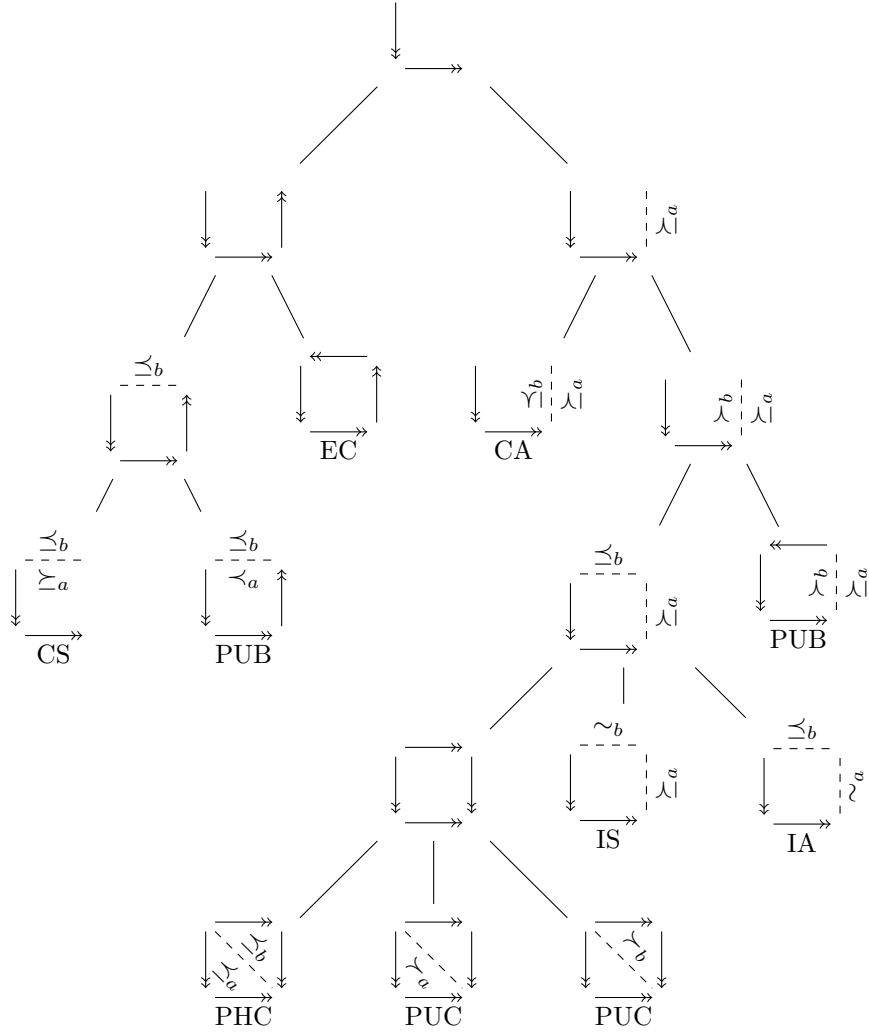
\begin{tikzpicture}[level distance=25mm]
\node{\begin{tikzpicture}[node distance=1cm]
		\node(x){};
		\node(y)[below of = x]{};
		\node(z)[right of = y]{};
		\node(t)[right of = x]{};
		\draw [->>] (x) -- (y);
		\draw [->>] (y) -- (z);
		\end{tikzpicture}
		}[sibling distance=50mm]
			child{node{\begin{tikzpicture}[node distance=1cm]
							\node(x){};
							\node(y)[below of = x]{};
							\node(z)[right of = y]{};
							\node(t)[right of = x]{};
							\draw [->>] (x) -- (y);
							\draw [->>] (y) -- (z);
							\draw [->>] (z) -- (t);
							\end{tikzpicture}
							}[sibling distance=25mm]
						child{node{\begin{tikzpicture}[node distance=1cm]
										\node(x){};
										\node(y)[below of = x]{};
										\node(z)[right of = y]{};
										\node(t)[right of = x]{};
										\draw [->>] (x) -- (y);
										\draw [->>] (y) -- (z);
										\draw [->>] (z) -- (t);
										\draw [dashed] (x) to node[above, midway] {$\preceq_b$} (t);
										\end{tikzpicture}
										}[sibling distance=25mm]
									child{node{\begin{tikzpicture}[node distance=1cm]
													\node(x){};
													\node(y)[below of = x]{};
													\node(z)[right of = y]{};
													\node(t)[right of = x]{};
													\draw [->>] (x) -- (y);
													\draw [->>] (y) to node[below] {CS} (z);
													\draw [dashed] (x) to node[above, midway] {$\preceq_b$} node[below, midway] {$\succeq_a$} (t);
													\end{tikzpicture}
													}
									}
									child{node{\begin{tikzpicture}[node distance=1cm]
													\node(x){};
													\node(y)[below of = x]{};
													\node(z)[right of = y]{};
													\node(t)[right of = x]{};
													\draw [->>] (x) -- (y);
													\draw [->>] (y) to node[below] {PUB} (z);
													\draw [->>] (z) -- (t);
													\draw [dashed] (x) to node[above, midway] {$\preceq_b$} node[below, midway] {$\prec_a$} (t);
													\end{tikzpicture}
													}
									}
						}
						child{node{\begin{tikzpicture}[node distance=1cm]
										\node(x){};
										\node(y)[below of = x]{};
										\node(z)[right of = y]{};
										\node(t)[right of = x]{};
										\draw [->>] (x) -- (y);
										\draw [->>] (y) to node[below] {EC} (z);
										\draw [->>] (z) -- (t);
										\draw [->>] (t) -- (x);
										\end{tikzpicture}
										}
						}
			}
			child{node{\begin{tikzpicture}[node distance=1cm]
							\node(x){};
							\node(y)[below of = x]{};
							\node(z)[right of = y]{};
							\node(t)[right of = x]{};
							\draw [->>] (x) -- (y);
							\draw [->>] (y) -- (z);
							\draw [dashed] (t) to node[below, midway, rotate=90] {$\succeq_a$} (z);
							\end{tikzpicture}
							}[sibling distance=25mm]
						child{node{\begin{tikzpicture}[node distance=1cm]
										\node(x){};
										\node(y)[below of = x]{};
										\node(z)[right of = y]{};
										\node(t)[right of = x]{};
										\draw [->>] (x) -- (y);
										\draw [->>] (y) to node[below] {CA} (z);
										\draw [dashed] (t) to node[below, midway, rotate=90] {$\succeq_a$} node[above, midway, rotate=90] {$\preceq_b$}(z);
										\end{tikzpicture}
										}[sibling distance=25mm]
						}
						child{node{\begin{tikzpicture}[node distance=1cm]
										\node(x){};
										\node(y)[below of = x]{};
										\node(z)[right of = y]{};
										\node(t)[right of = x]{};
										\draw [->>] (x) -- (y);
										\draw [->>] (y) -- (z);
										\draw [dashed] (t) to node[below, midway, rotate=90] {$\succeq_a$} node[above, midway, rotate=90] {$\succ_b$}(z);
										\end{tikzpicture}
										}[sibling distance=25mm]
									child{node{\begin{tikzpicture}[node distance=1cm]
													\node(x){};
													\node(y)[below of = x]{};
													\node(z)[right of = y]{};
													\node(t)[right of = x]{};
													\draw [->>] (x) -- (y);
													\draw [->>] (y) -- (z);
													\draw [dashed] (t) to node[below, midway, rotate=90] {$\succeq_a$} (z);
													\draw [dashed] (x) to node[above, midway] {$\preceq_b$} (t);
													\end{tikzpicture}
													}
												child{node{\begin{tikzpicture}[node distance=1cm]
																\node(x){};
																\node(y)[below of = x]{};
																\node(z)[right of = y]{};
																\node(t)[right of = x]{};
																\draw [->>] (x) -- (y);
																\draw [->>] (y) -- (z);
																\draw [->>] (t) -- (z);
																\draw [->>] (x) -- (t);
																\end{tikzpicture}
																}
															child{node{\begin{tikzpicture}[node distance=1cm]
																			\node(x){};
																			\node(y)[below of = x]{};
																			\node(z)[right of = y]{};
																			\node(t)[right of = x]{};
																			\draw [->>] (x) -- (y);
																			\draw [->>] (y) to node[below] {PHC} (z);
																			\draw [->>] (t) -- (z);
																			\draw [->>] (x) -- (t);
																			\draw [dashed] (x) to node[below, midway, rotate=-45] {$\preceq_a$} node[above, midway, rotate=-45] {$\preceq_b$}(z);
																			\end{tikzpicture}
																			}
															}
															child{node{\begin{tikzpicture}[node distance=1cm]
																			\node(x){};
																			\node(y)[below of = x]{};
																			\node(z)[right of = y]{};
																			\node(t)[right of = x]{};
																			\draw [->>] (x) -- (y);
																			\draw [->>] (y) to node[below] {PUC} (z);
																			\draw [->>] (t) -- (z);
																			\draw [->>] (x) -- (t);
																			\draw [dashed] (x) to node[below, midway, rotate=-45] {$\succ_a$} (z);
																			\end{tikzpicture}
																			}
															}
															child{node{\begin{tikzpicture}[node distance=1cm]
																			\node(x){};
																			\node(y)[below of = x]{};
																			\node(z)[right of = y]{};
																			\node(t)[right of = x]{};
																			\draw [->>] (x) -- (y);
																			\draw [->>] (y) to node[below] {PUC} (z);
																			\draw [->>] (t) -- (z);
																			\draw [->>] (x) -- (t);
																			\draw [dashed] (x) to node[above, midway, rotate=-45] {$\succ_b$}(z);
																			\end{tikzpicture}
																			}
															}
												}
												child{node{\begin{tikzpicture}[node distance=1cm]
																\node(x){};
																\node(y)[below of = x]{};
																\node(z)[right of = y]{};
																\node(t)[right of = x]{};
																\draw [->>] (x) -- (y);
																\draw [->>] (y) to node[below] {IS} (z);
																\draw [dashed] (t) to node[below, midway, rotate=90] {$\succeq_a$} (z);
																\draw [dashed] (x) to node[above, midway] {$\sim_b$} (t);
																\end{tikzpicture}
																}
												}
												child{node{\begin{tikzpicture}[node distance=1cm]
																\node(x){};
																\node(y)[below of = x]{};
																\node(z)[right of = y]{};
																\node(t)[right of = x]{};
																\draw [->>] (x) -- (y);
																\draw [->>] (y) to node[below] {IA} (z);
																\draw [dashed] (t) to node[below, midway, rotate=90] {$\sim_a$} (z);
																\draw [dashed] (x) to node[above, midway] {$\preceq_b$} (t);
																\end{tikzpicture}
																}
												}
									}
									child{node{\begin{tikzpicture}[node distance=1cm]
													\node(x){};
													\node(y)[below of = x]{};
													\node(z)[right of = y]{};
													\node(t)[right of = x]{};
													\draw [->>] (x) -- (y);
													\draw [->>] (y) to node[below] {PUB} (z);
													\draw [dashed] (t) to node[below, midway, rotate=90] {$\succeq_a$} node[above, midway, rotate=90] {$\succ_b$}(z);
													\draw [->>] (t) -- (x);
													\end{tikzpicture}
													}
									}
						}
			};
\end{tikzpicture}
\end{proof}

\begin{proof}[Proof of Lemma~\ref{lem:strong-uniform-continuity}]
The function $f$ is continuous on a compact, so it is bounded, so wlog let us assume that its range is included in $[0,1]^2$. Let $I_k \,:=\, ]\frac{\epsilon(k-1)}{2},\frac{\epsilon(k+1)}{2}[$ for every natural number $k \leq 2\cdot\epsilon^{-1}$. By continuity every $f^{-1}[I_{l(1)}\times \dots \times I_{l(n)}]$ is open so it can be written as a union $\cup A_l$ of Cartesian products of open subsets of the $S_i$. Since the union $\cup_{l\in\{1,\dots,\lfloor 2\cdot\epsilon^{-1}\rfloor\}^n} A_l$ constitutes an open cover of $S$, by compactness it has a finite open subcover $\{B_{1j} \times\dots\times B_{nj}\,\mid\, 1\leq j\leq q\}$, and each $B_{1j} \times\dots\times B_{nj}$ must belong to at least one $A_l$. For every $i \leq n$ and $u\in\{0,1\}^q$ let $B_i^u \,:=\, \cap^q_{j=1} B^u_{ij}$ where $B^u_{ij} := B_{ij}$ if $u_j = 1$ and $B^u_{ij}:=X_i\backslash B_{ij}$ otherwise. By construction for all $u_1,\dots,u_n \in \{0,1\}^q$ the set $B^{u_1}_1\times\dots\times B^{u_n}_n$ belongs to at least one $A_l$, so for all $x,y\in B^{u_1}_1\times\dots\times B^{u_n}_n$, we have $|f(x)-f(y)| < \epsilon$ by definition of the intervals $I_k$. Let us finally delete the empty $B_i^u$ from their respective families and split some remaining $B^u_i$ to obtain $n$ partitions of a common cardinality $m$. 
\end{proof}

\begin{proof}[Proof of Lemma~\ref{lem:ind-fs}]
Let $A=(a_{ij})_{1\leq i\leq n\,\wedge 1\leq j\leq m}$ be a matrix and let $x$ occurs in $A$. Let $\preceq$ be the binary relation defined over the (indexes of the) columns of the matrix $A$ by $j\preceq j':=\forall i,\, a_{ij'}=x\Rightarrow a_{ij}=x$. The relation is a preorder by definition, and if moreover $A$ has no forbidden rectangle, it is total: if $a_{ij}=x$ and $a_{ij'}\neq x$ for some $i$, $j$, and $j'$, absence of forbidden rectangles implies that $a_{i'j'}= x\Rightarrow a_{i'j}= x$ for all $i'$, which shows that $\preceq$ is total. So let $\theta$ be a permutation of $\{1,\dots,m\}$ such that $\theta(j)\preceq\theta(j+1)$ for all $j<m$. Likewise, let $\varphi$ be a permutation of $\{1,\dots,n\}$ that rearranges the lines, and let define the matrix $B$ by $b_{ij}:=a_{\varphi(i)\theta(j)}$ for all $i$ and $j$. For $1\leq j\leq m$ let $s(j)$ be the number of occurrences of the entry $x$ in the $j$-th column of $B$, so that $1\leq s(1)$ since $x$ occurs in $A$. Graphically, the "$x$-area" of $B$ is the exactly the area above the staircase defined by $s$, as in the pictures above and below.

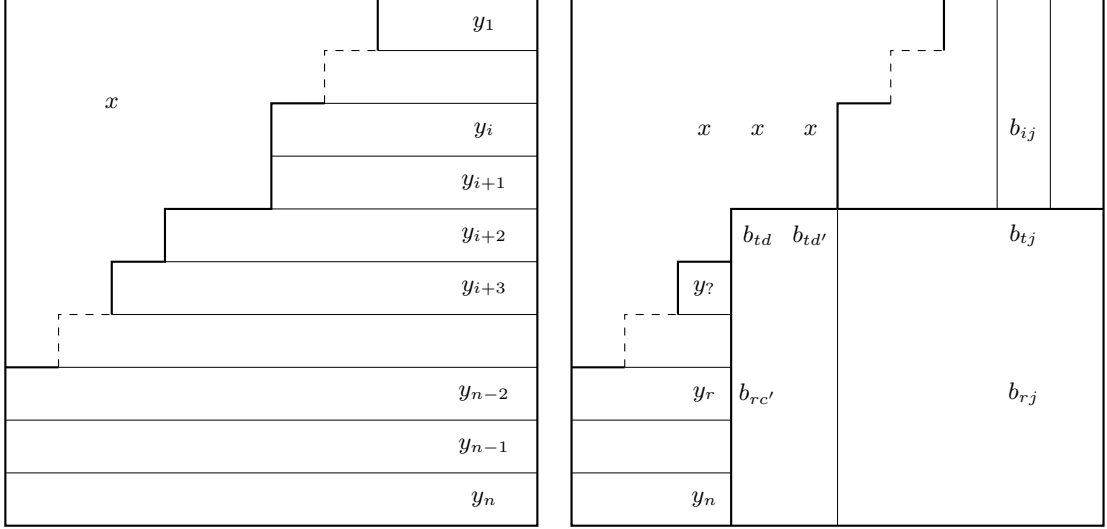
\begin{tabular}{cc}
\begin{tikzpicture}[scale=.7]
\draw[thick] (0,0) rectangle (10,10);

\draw[thick] (0,3) -- (1,3);
\draw[dashed] (1,3) -- (1,4) -- (2,4);
\draw[thick] (2,4) -- (2,5) -- (3,5) -- (3,6) -- (5,6) -- (5,8) -- (6,8);
\draw[dashed] (6,8) -- (6,9) -- (7,9);
\draw[thick] (7,9) -- (7,10);

\node at (2,8) {$x$};
\node at (9,.5) {$y_n$};
\node at (9,1.5) {$y_{n-1}$};
\node at (9,2.5) {$y_{n-2}$};
\node at (9,4.5) {$y_{i+3}$};
\node at (9,5.5) {$y_{i+2}$};
\node at (9,6.5) {$y_{i+1}$};
\node at (9,7.5) {$y_{i}$};
\node at (9,9.5) {$y_1$};

\draw[very thin] (0,1) -- (10,1);
\draw[very thin] (0,2) -- (10,2);
\draw[very thin] (1,3) -- (10,3);
\draw[very thin] (2,4) -- (10,4);
\draw[very thin] (3,5) -- (10,5);
\draw[very thin] (5,6) -- (10,6);
\draw[very thin] (5,7) -- (10,7);
\draw[very thin] (5,8) -- (10,8);
\draw[very thin] (7,9) -- (10,9);
\end{tikzpicture}
&
\begin{tikzpicture}[scale=.7]
\draw[thick] (0,0) rectangle (10,10);

\draw[thick] (0,3) -- (1,3);
\draw[dashed] (1,3) -- (1,4) -- (2,4);
\draw[thick] (2,4) -- (2,5) -- (3,5) -- (3,6) -- (5,6) -- (5,8) -- (6,8);
\draw[dashed] (6,8) -- (6,9) -- (7,9);
\draw[thick] (7,9) -- (7,10);

\node at (2.5,.5) {$y_n$};
\node at (2.5,2.5) {$y_r$};
\node at (3.5,2.5) {$b_{rc'}$};
\node at (2.5,4.5) {$y_?$};
\node at (3.5,5.5) {$b_{td}$};
\node at (4.5,5.5) {$b_{td'}$};

\node at (8.5,7.5) {$b_{ij}$};
\node at (2.5,7.5) {$x$};
\node at (3.5,7.5) {$x$};
\node at (4.5,7.5) {$x$};
\node at (8.5,2.5) {$b_{rj}$};
\node at (8.5,5.5) {$b_{tj}$};

\draw[thick] (3,5) -- (3,0);
\draw[very thin] (5,6) -- (5,0);
\draw[very thin] (8,10) -- (8,6);
\draw[very thin] (9,10) -- (9,6);

\draw[very thin] (0,1) -- (3,1);
\draw[very thin] (0,2) -- (3,2);
\draw[very thin] (1,3) -- (3,3);
\draw[very thin] (2,4) -- (3,4);
\draw[thick] (5,6) -- (10,6);
\end{tikzpicture}
\end{tabular}
If the complement of the $x$-area is made of horizontal constant-valued stripes only, as is shown in the left-hand picture above, let $k:=m$ and the statement of the lemma is satisfied. Otherwise, let $k$ be the maximal column index such that $\forall i,j,j',\,(j\leq j'\leq k\wedge s(j)<i)\Rightarrow b_{ij}=b_{ij'}$ (so that $k<m$ due to the case disjunction) and such that $s(k+1)<s(k)$. The column of index $k$ is the one displaying $x$, $y_?$, $y_r$, and $y_n$ in the right-hand picture above. By definition of $k$, Assertions~\ref{ind-fs:cond1} and \ref{ind-fs:cond3} are satisfied.

If $s(m)=s(k+1)$, Assertion~\ref{ind-fs:cond2} is satisfied too; otherwise, by definition of $k$ there exist a row index $r$ and column indexes $c$ and $c'$ such that $s(k+1)=s(c')\leq s(c)<r$ and $b_{rc}\neq b_{rc'}$. The right-hand picture above displays the case where $c\leq k$ (so that one may assume $c=k$) and $b_{rk}= y_{r}$, and also the case $k<c$ by using symbols $t,d,d'$ instead of $r,c,c'$. Now let $j$ be a column index such that $s(j)<s(k+1)$ and let $i$ be a row index such that $i\leq s(k+1)$, so that $b_{ic'}=b_{ic}=x$ ($b_{ic}=b_{ik}$ in the picture) and $b_{ij}\neq x$. Since $b_{rc}\neq b_{rc'}$ (or $b_{rc}\neq y_r$ in the picture) by assumption and absence of forbidden rectangles, $b_{ij}=b_{rj}$. Since this holds for every $i\leq s(k+1)$, this proves Assertion~\ref{ind-fs:cond2} about vertical constant-valued stripes. 

Assertion~\ref{ind-fs:cond4} holds since absence of forbidden rectangles is preserved by submatrix. As for Assertion~\ref{ind-fs:cond5}, let $s(k)<i$ and $s(j)<s(k+1)$, and notice that $b_{s(k)k}=x$ and $b_{ik}\neq x$ and $b_{s(k)j}\neq x$, so $b_{ij}=b_{s(k)+1,j}$ or $b_{ij}=b_{ik}$ by absence of forbidden rectangles.

\end{proof}

\begin{proof}[Proof of Lemma~\ref{lem:or}]
Let us proceed by induction on the number $N$ of different elements of $X$ that are involved in $A$. The claim holds for $N=1$, so let $1<N$ and let us assume that the property holds for $N-1$. Lemma~\ref{lem:ind-fs} states that, up to row and column permutation, $A$ is as the left-hand picture below. Let $k$ be as in Lemma~\ref{lem:ind-fs}, that is, the index of the column just to the left of the submatrix $B$ and let $l$ be the index of the row just above $B$. Note that $k$, $l$, $r$, and $c$ could equal $0$, $0$, $n+1$, and $m+1$, respectively. 

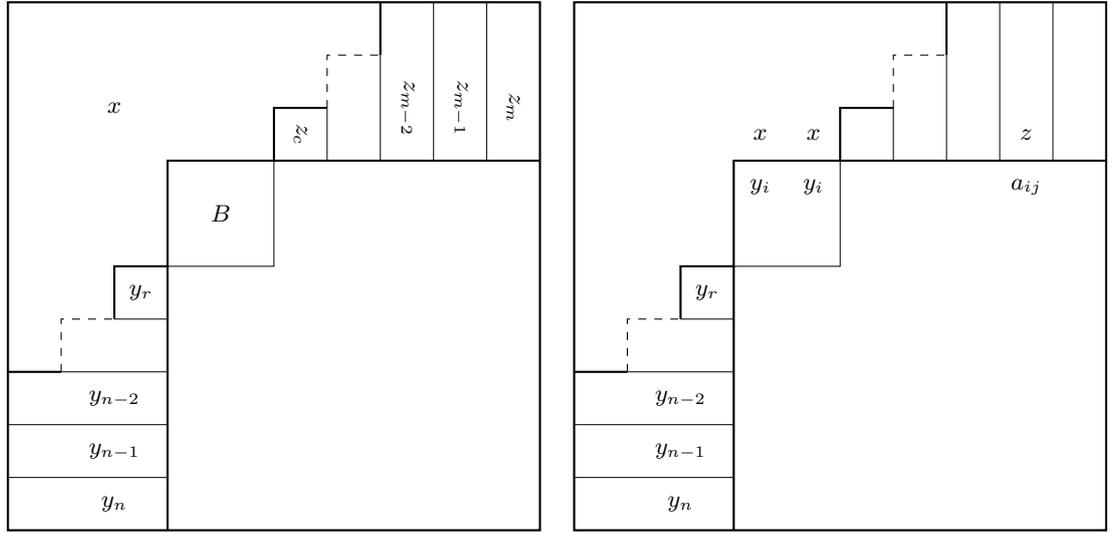
\begin{tabular}{cc}
\begin{tikzpicture}[scale=.7]
\draw[thick] (0,0) rectangle (10,10);

\draw[thick] (0,3) -- (1,3);
\draw[dashed] (1,3) -- (1,4) -- (2,4);
\draw[thick] (2,4) -- (2,5) -- (3,5) -- (3,7) -- (5,7) -- (5,8) -- (6,8);
\draw[dashed] (6,8) -- (6,9) -- (7,9);
\draw[thick] (7,9) -- (7,10);

\node at (2,8) {$x$};
\node at (2,.5) {$y_n$};
\node at (2,1.5) {$y_{n-1}$};
\node at (2,2.5) {$y_{n-2}$};
\node at (2.5,4.5) {$y_r$};

\node at (5.5,7.5) [rotate = -90] {$z_c$};
\node at (7.5,8) [rotate = -90] {$z_{m-2}$};
\node at (8.5,8) [rotate = -90] {$z_{m-1}$};
\node at (9.5,8) [rotate = -90] {$z_m$};

\draw[thick] (3,5) -- (3,0);
\draw[very thin] (6,8) -- (6,7);
\draw[very thin] (7,9) -- (7,7);
\draw[very thin] (8,10) -- (8,7);
\draw[very thin] (9,10) -- (9,7);

\draw[very thin] (0,1) -- (3,1);
\draw[very thin] (0,2) -- (3,2);
\draw[very thin] (1,3) -- (3,3);
\draw[very thin] (2,4) -- (3,4);
\draw[thick] (5,7) -- (10,7);

\draw[very thin] (3,5) -- (5,5) -- (5,7);
\node at (4,6) {$B$};
\end{tikzpicture}
&
\begin{tikzpicture}[scale=.7]
\draw[thick] (0,0) rectangle (10,10);

\draw[thick] (0,3) -- (1,3);
\draw[dashed] (1,3) -- (1,4) -- (2,4);
\draw[thick] (2,4) -- (2,5) -- (3,5) -- (3,7) -- (5,7) -- (5,8) -- (6,8);
\draw[dashed] (6,8) -- (6,9) -- (7,9);
\draw[thick] (7,9) -- (7,10);

\node at (2,.5) {$y_n$};
\node at (2,1.5) {$y_{n-1}$};
\node at (2,2.5) {$y_{n-2}$};
\node at (2.5,4.5) {$y_r$};

\node at (3.5,7.5) {$x$};
\node at (4.5,7.5) {$x$};
\node at (8.5,7.5) {$z$};
\node at (3.5,6.5) {$y_i$};
\node at (4.5,6.5) {$y_i$};
\node at (8.5,6.5) {$a_{ij}$};

\draw[thick] (3,5) -- (3,0);
\draw[very thin] (6,8) -- (6,7);
\draw[very thin] (7,9) -- (7,7);
\draw[very thin] (8,10) -- (8,7);
\draw[very thin] (9,10) -- (9,7);

\draw[very thin] (0,1) -- (3,1);
\draw[very thin] (0,2) -- (3,2);
\draw[very thin] (1,3) -- (3,3);
\draw[very thin] (2,4) -- (3,4);
\draw[thick] (5,7) -- (10,7);

\draw[very thin] (3,5) -- (5,5) -- (5,7);
\end{tikzpicture}
\end{tabular}

As a submatrix of a matrix without forbidden rectangles, $B$ has no forbidden rectangle either, and it involves at most $N-1$ different elements from $X$, so by induction hypothesis let $(y_{l+1},\dots,y_{r-1})\in X^{r-l-1}$ and $(z_{k+1},\dots,z_{c-1})\in X^{c-k-1}$ such that $a_{i,j}=y_i$ or $a_{ij}=z_j$ whenever $l< i<r$ and $k<j<c$. For all $l< i<r$, if $a_{ij}=a_{i,k+1}$ for all $k<j<c$, let us replace $y_i$ with $a_{i,k+1}$; likewise, for all $k<j<c$, if $a_{ij}=a_{l+1,j}$ for all $l<i<r$, let us replace $z_j$ with $a_{l+1,j}$. Note that these changes preserve the property that $a_{ij}=y_i$ or $a_{ij}=z_j$ whenever $l< i<r$ and $k<j<c$, and let us show that $(x,\dots,x,y_{l+1},\dots y_n)$ and $(x,\dots,x,z_{k+1},\dots z_m)$ witness the claim. The only problematic case is when $l<i<r$ and $c\leq j$ (and its symmetric case when $k<j<c$ and $r\leq i$), so let $l<i<r$ and $c\leq j$, as in the right-hand picture above. If $a_{ij}=z_j$, this proves the claim, so let us assume that $a_{ij}\neq z_j$. For all $k<j'<c$, the absence of forbidden rectangle implies that $a_{ij'}=a_{ij}$, so $a_{ij}=y_i$ by the replacement performed above.  
\end{proof}

\begin{proof}[Proof of Theorem~\ref{thm:fr-d}]
By induction on the size of the game structure. Let $G$ be a non-stripe-decomposable game structure, \textit{i.e.} all rows and columns of $G$ involve at least two outcomes, and let us prove that $G$ must be corner-decomposable. Let us assume without loss of generality that $G$ has at least four rows and tree columns. Let $G'$ be $G$ minus its lowest row. Since $G'$ has not forbidden rectangle, it is stripe-decomposable or corner-decomposable by induction hypothesis. Let us make a case disjunction.

First case, $G'$ is stripe-decomposable, and since $G$ is not stripe-decomposable, $G'$ must be stripe-decomposable along a column, so that up to column permutation, $G$ looks like the left-hand picture picture below.

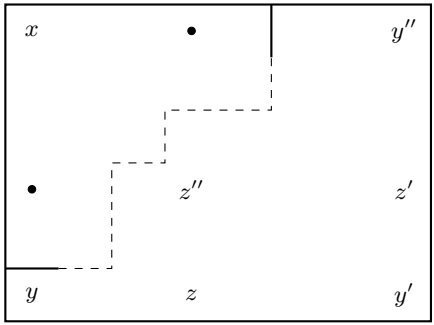
\begin{tabular}{c@{\hspace{1cm}}c}
\begin{tikzpicture}[scale=.7]
\draw[thick] (0,0) rectangle (8,6);

\draw[thick] (0,1) -- (1,1);
\draw[dashed] (1,1) -- (2,1) -- (2,3) -- (3,3) -- (3,4) -- (5,4) -- (5,5) ;
\draw[thick] (5,5) -- (5,6);

\node at (.5,5.5) {$x$};
\node at (.5,.5) {$y$};
\node at (3.5,.5) {$z$};
\node at (7.5,.5) {$y'$};
\node at (7.5,2.5) {$z'$};
\node at (7.5,5.5) {$y''$};
\node at (3.5,2.5) {$z''$};
\filldraw [black](.5,2.5) circle (2pt);
\filldraw [black](3.5,5.5) circle (2pt);
\end{tikzpicture}
&
$\begin{array}{|c|c|c|ccc|}
	  \cline{1-6}
		x & x & z &  & x_2 & \\
	  \cline{1-6}
		x & y & y &  & x_1 & \\
	  \cline{1-6}
		z & y & z &  & z' & \\
	  \cline{1-6}
		 &  &  & \multicolumn{3}{c|}{}\\
  	  &  &  & \multicolumn{3}{c|}{}\\
		 &  &  & \multicolumn{3}{c|}{}\\
	  \cline{1-6}
		y_0 & y_1 & y_2 & & x_0 &\\
	  \cline{1-6}
\end{array}$
\end{tabular}

The $x$ staircase area is obtained by permutation of rows (but the lowest row) and columns, as in Lemma~\ref{lem:ind-fs}. To avoid stripe-decomposition, the outcome $y$ must be different from $x$ and the $x$ area must not cover the rightmost part of the highest row. Moreover the bottom-right outcome $y'$ equals $y$, otherwise the rightmost column would be filled with $y'$ to avoid forbidden rectangles, which would in turn imply that $G$ is stripe-decomposable, thus contradicting the assumption. Again to avoid stripe-decomposition, some outcomes $z$ and $z'$ distinct from $y$ must occur in the lowest row and rightmost column, respectively. Subsequently the top-right outcome $y''$ must equal $y$ to avoid forbidden rectangles that would involve $z$. Again to avoid stripe-decomposition, a $z'$ distinct from $y$ must occur in the right-most column. Let $z''$ be the outcome at the intersection of the column of $z$ and row of $z'$. Again to avoid forbidden rectangles, $z''$ must equal $z$ or $z'$, which means that $z''$ is distinct from $x$. (That is why $z''$ was drawn outside of the $x$ area in the first place.) Moreover $z''=z'$ to avoid a forbidden rectangle involving $z''$ and $y''$, and likewise  $z''=z$ to avoid a forbidden rectangle involving $z''$ and $y'$, so $z = z' = z''$. The substructure made of the displayed outcomes plus the two black spots corresponds (by swapping $y$ and $z$) to the top-left corner of the corner-decomposition pattern. By invoking another six times the absence of forbidden rectangles, one shows that the requirements from Definition~\ref{def:scd} on $x/z$, $x/y$, and $y/z$ hold.

Second case, $G'$ is corner-decomposable, so $G$ looks like in the right-hand picture above up to row and column permutations. If $(y_0,y_1,y_2) \in \{x,z\} \times \{x,y\} \times \{y,z\}$ then $G$ is corner-decomposable, otherwise let us assume without loss of generality that $y_0\notin\{x,z\}$. The rectangle that involves $y_0$, $x$, $z$, and $y_2$ shows that $y_0 = y_2$, then the rectangle that involves $y_1$, $x$, $z$, and $y_2$ shows that $y_1 = y_2$. Since $G$ is not stripe-decomposable some $x_0$ distinct from $y_0$ must occur in the lowest row. Let $x_2$ be the outcome occurring in the same column as $x_0$ and in the top row (or more precisely the row among the top three rows that involves only $x$ and $z$, from which $y_0$ was assumed distinct). Since $x_2\in\{x,z\}$ but is distinct from one of then, the rectangles involving $y_1$ and $x_2$ on the one hand, and $y_2$ and $x_2$ on the other hand shows that $x_0 = x_2$. Let us make a nested case disjunction. First subcase, if $x_0 = z$, the rectangle that involves $x_1$, $x_2$, $y$, and $z$ shows that $x_1 = y$, so that the rectangle involving $y_0$ and $x_1$ is a forbidden rectangle. Second subcase, $x_0 = x$, so that the rectangle that involves $y_0$ and $x_1$ shows that $x_1 = x$, and the rectangle that involves $y_2$ and $z'$ shows that $z' = z$, and subsequently the rectangle that involves $z'$ and $y_1$ shows that $y_1 = y$. Finally the substructure that involves $y_1$, $y_2$, $x_0$, $z'$, and $x_2$ (and the four remaining $x$, $y$, and $z$ twice) corresponds to the top-left corner of the corner-decomposition pattern.
\end{proof}






\end{document}